\documentclass[showpacs,amsmath,amssymb,twocolumn,superscriptaddress,notitlepage,preprintnumbers,pra]{revtex4-1}

\usepackage{graphicx}
\usepackage{dcolumn}
\usepackage{bm}
\usepackage{times}
\usepackage{multirow}
\usepackage{dcolumn}
\usepackage{tabularx}
\usepackage{subfigure}
\usepackage{ae}
\usepackage{lettrine}
\usepackage{color}
\usepackage{amsmath,amsthm,amssymb,amsfonts,boxedminipage}
\usepackage{epstopdf}
\usepackage{hyperref}
\usepackage{dsfont}
\newcommand{\algmargin}{\the\ALG@thistlm}
\makeatother

\newtheorem{theorem}{Theorem}
\newtheorem{lemma}{Lemma}
\newtheorem{corollary}{Corollary}
\newtheorem{definition}{Definition}

\newtheorem{fact}{Fact}

\makeatletter
\renewcommand\@biblabel[1]{#1.}
\makeatother

\usepackage[noend,noline,ruled,linesnumbered]{algorithm2e}
\usepackage{hyperref}
\usepackage{bm}
\usepackage{times}
\usepackage{multirow}
\usepackage{dcolumn}
\usepackage{tabularx}
\usepackage{graphicx}
\usepackage{subfigure}
\usepackage{ae}
\usepackage{lettrine}
\usepackage{color}
\usepackage{makecell}

\usepackage{amsmath,amsthm,amssymb,amsfonts,boxedminipage}
\usepackage{epstopdf}

\definecolor{comment}{rgb}{0, 0, 0}

\hypersetup{colorlinks=true, linkcolor=cyan, citecolor=cyan, urlcolor=blue }

\usepackage{braket}

\newcommand{\abs}[1]{\left| #1 \right|}

\newtheorem{problem}{Problem}

\newcommand{\PKU}{Center on Frontiers of Computing Studies, School of Computer Science, Peking University, Beijing 100871, China}
\newcommand{\bnu}{School of Artificial Intelligence,
 Beijing Normal University, Beijing,
 100875, China}

\begin{document}


\title{Hamiltonian Dynamics Learning: A Scalable Approach to Quantum Process Characterization}

\author{Yusen Wu}
\thanks{These authors contributed equally}
\affiliation{\bnu}

\author{Yukun Zhang}
\thanks{These authors contributed equally}
\affiliation{\PKU}

\author{Chuan Wang}
\email{wangchuan@bnu.edu.cn}
\affiliation{\bnu}

\author{Xiao Yuan}
\email{xiaoyuan@pku.edu.cn}
\affiliation{\PKU}

\date{\today}

\begin{abstract}
Quantum process characterization is a fundamental task in quantum information processing, yet conventional methods, such as quantum process tomography, require prohibitive resources and lack scalability. Here, we introduce an efficient quantum process learning method specifically designed for short-time Hamiltonian dynamics. Our approach reconstructs an equivalent quantum circuit representation from measurement data of unknown Hamiltonian evolution without requiring additional assumptions and achieves polynomial sample and computational efficiency.
Our results have broad applications in various directions. We demonstrate applications in quantum machine learning, where our protocol enables efficient training of variational quantum neural networks by directly learning unitary transformations. Additionally, it facilitates the prediction of quantum expectation values with provable efficiency and provides a robust framework for verifying quantum computations and benchmarking realistic noisy quantum hardware. This work establishes a new theoretical foundation for practical quantum dynamics learning, paving the way for scalable quantum process characterization in both near-term and fault-tolerant quantum computing.

\end{abstract}

\maketitle


Quantum computers are entering regimes beyond the reach of classical computational power~\cite{arute2019quantum,morvan2024phase,zhong2020quantum}. Coherent manipulation of complex quantum states with hundreds of physical qubits has been demonstrated across multiple platforms, including trapped ions~\cite{smith2016many}, neutral atom arrays~\cite{evered2023high}, and superconducting qubit circuits~\cite{arute2019quantum,morvan2024phase,acharya2024quantum}. As quantum hardware continues to scale in size and complexity, the ability to characterize quantum processes becomes critical for advancing quantum error correction code~\cite{bravyi2024high, acharya2024quantum,putterman2025hardware}, quantum error mitigation~\cite{kim2023scalable,o2023purification}, and quantum algorithms~\cite{Youngseok2023evidence,morvan2024phase,huang2022quantum,wu2023quantum,guo2024experimental}. Among characterization tools, quantum process tomography (QPT) stands as a foundational method for reconstructing unknown quantum processes from measurement data, thereby elucidating the intrinsic structure of quantum circuits~\cite{mohseni2008quantum, gebhart2023learning,ahmed2023gradient}. Predicting the expectation value of the output of quantum processes has been proved to be efficient in some cases~\cite{raza2024online,huang2023learning2,mcginley2024postselection}, however, general QPT requires an exponential query complexity in the worst-case scenario~\cite{haah2023query}, rendering it infeasible for large-scale systems.

To address this challenge, various heuristic approaches have been developed, leveraging parameterized quantum ansätze~\cite{xue2022variational,xue2023variational,wulearning}, Bayesian inference~\cite{granade2012robust, wiebe2014hamiltonian,stenberg2014efficient}, neural network models~\cite{xin2019local, melko2019restricted}, and tensor networks~\cite{torlai2023quantum}. However, neural network and tensor network methods generally lack theoretical guarantees, making their applicability and scalability uncertain.  
Notably, regarding shallow quantum circuits with constant circuit depths, Huang et al.~\cite{huang2024learning} recently demonstrated the feasibility of efficient process characterization, making it a promising approach for analyzing large-scale quantum circuits.
Nevertheless, for more general quantum processes governed by Hamiltonian dynamics, such as those in both analog quantum simulators and digital quantum computers, the problem remains an open challenge.

Recent advances in provable Hamiltonian learning protocols have partially addressed this challenge by focusing on learning the underlying Hamiltonian. These approaches generally assume prior knowledge of the Hamiltonian structure~\cite{huang2023learning,gu2024practical}, though recent works have relaxed this assumption~\cite{haah2024learning,zhao2024learning, bakshi24structure,ma2024learning} and extended to general sparse Hamiltonians~\cite{hu2025ansatz}. Additionally, they typically require access to tunable evolution times \( t \) for Hamiltonian dynamics \( e^{-iHt} \) or a known inverse temperature \( \beta \) for thermal states of the form \( e^{-\beta H}/{\rm Tr}[e^{-\beta H}] \)~\cite{anshu2021sample,haah2024learning}.  
Since the Hamiltonian and its corresponding dynamics are generally nontrivially related, these results cannot be directly applied to the problem of learning Hamiltonian dynamics. This naturally raises an intriguing question: 
\emph{``Can we efficiently learn Hamiltonian dynamics without prior structural knowledge and without access to tunable evolution times?"}

\begin{figure*}
\centering
\includegraphics[width=\textwidth]{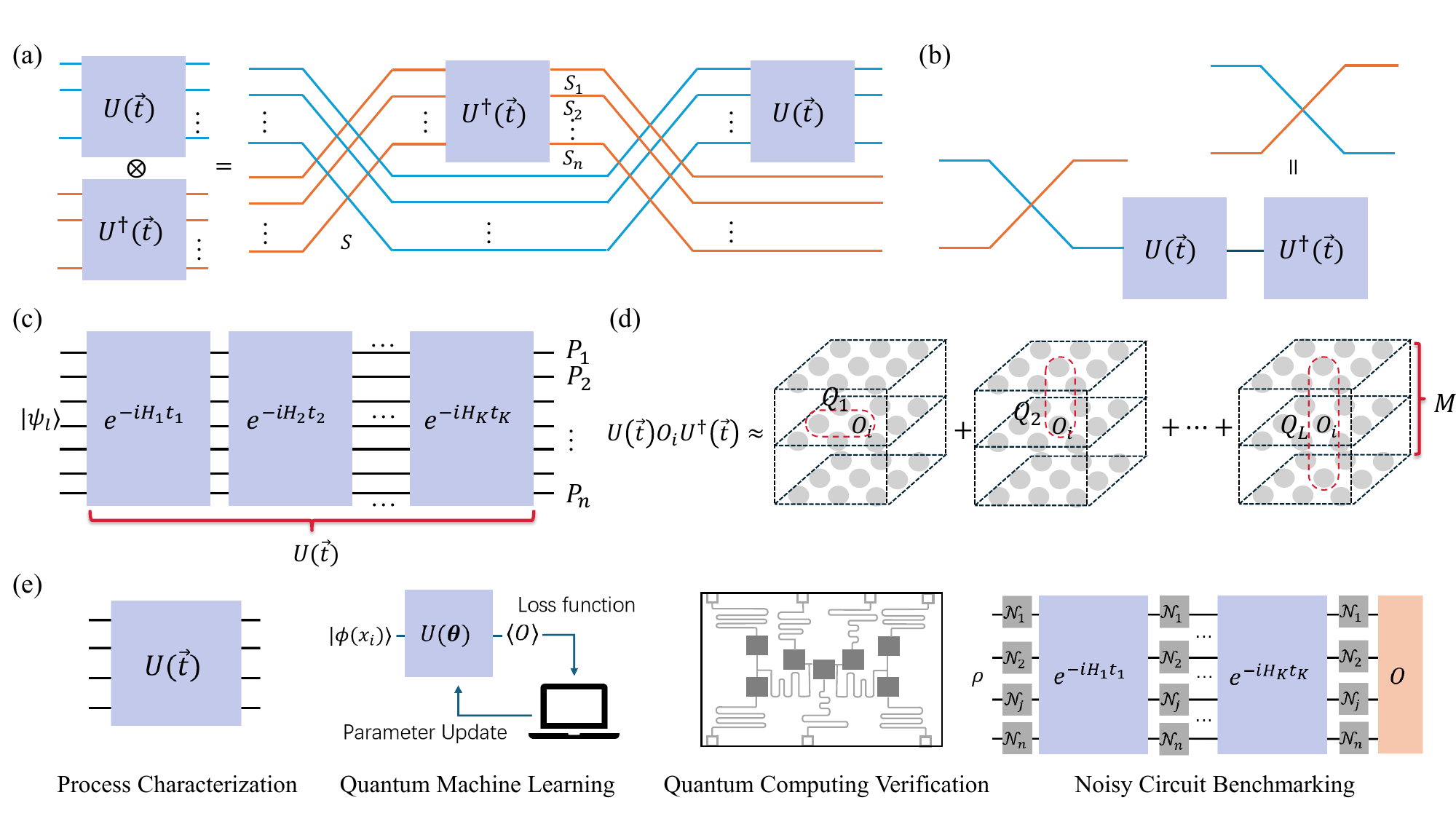} 
\caption{(a),~(b) A basic idea to prove the identity Eq.~\eqref{Eq:Identity}. Specifically, (a) demonstrates the relationship $U(\vec{t})\otimes U(\vec{t})=SU^{\dagger}(\vec{t})S_1\cdots S_nU(\vec{t})$, meanwhile (b) inserts an identity $U(\vec{t})U^{\dagger}(\vec{t})$ between swap operator pairs $S_i,S_{i+1}$, which thus gives rise to Eq.~\eqref{Eq:Identity}. (c) In the learning phase, random tensor product state $|\psi_l\rangle$ is applied by the unknown quantum process $U(\vec{t})$, followed by a random Pauli measurement. (d) Visualization on approximating $U(\vec{t})O_iU^{\dagger}(\vec{t})\approx\sum_{Q(O_i)}\alpha_{Q(O_i)}Q(O_i)$ by using the cluster expansion method via Lemma~\ref{lemma1}. In this visualization, each grey point represents a single qubit and the red dot circle represents a Pauli term $Q$. 
  The approximation is essentially a linear combination of ${\rm poly}(n)$ matrices induced by connected clusters, and it is applied to $\max\{\abs{{\rm supp}(Q(O_i))}\}\leq M$ qubits. (e) The proposed quantum learning algorithm can be employed to characterize an unknown quantum process, train HV ansatz based quantum machine learning models, verify the output of quantum computers, and benchmark quantum states produced by weakly noisy quantum circuits.  }
\label{fig:1}
\end{figure*}

Here, we address this problem by presenting an efficient protocol for learning short-time Hamiltonian dynamics. Specifically, we consider an unknown unitary process  
$U(\vec{t}) = e^{-iH^{(1)}t_1} \cdots e^{-iH^{(K)}t_K}$, 
which is generated by a sequence of \( n \)-qubit, \(\mathcal{O}(1)\)-dimensional local Hamiltonians \(\{H^{(k)}\}_{k=1}^K\) and a time series \(\vec{t} = \{t_k\}_{k=1}^K\).  
We propose an efficient protocol to learn an approximation and construct a unitary that is \(\epsilon\)-close to \(U(\vec{t})\) in terms of the diamond norm. Our results have broad applications in quantum machine learning, quantum computation verification, and noisy device benchmarking.  
First, they provide an efficient method for training Hamiltonian variational ansatz-based quantum neural networks (QNNs) for classification tasks. While optimizing variational parameters is known to be NP-hard~\cite{bittel2021training}, our approach circumvents this difficulty by directly learning the QNN’s unitary action without requiring parameter extraction.  
Next, the learned model can be used to predict quantum expectation values. Specifically, for global observables with two-dimensional Hamiltonian dynamics, predictions require quasi-polynomial classical time, whereas for local observables with constant-dimensional dynamics, efficient classical prediction is feasible. 
Finally, our protocol extends to noisy quantum devices, offering an efficient approach for benchmarking realistic large-scale quantum processes.

\vspace{0.2cm}

\noindent\textbf{Hamiltonian dynamics learning} --- We first introduce the definition of Hamiltonian dynamics learning.
We consider an $n$-qubit $D$-dimensional geometrically local Hamiltonian
\begin{align}
    H=\sum\limits_{X\in S}\lambda_Xh_X,
    \label{Eq:Hamiltonian}
\end{align}
where $S$ represents a set of subsystems, real-valued coefficient $\abs{\lambda_X}\leq 1$, and $h_X$ represents a Hermitian operator non-trivially acting on the local qubit set $X\subset S$. Without loss of generality, we assume the operator norm of each $h_X$ satisfies $\|h_X\|\leq 1$, meanwhile their locality satisfies $\max_{X\subset S}\abs{{\rm supp}(h_X)}=\Lambda$.
To characterize the locality and correlations presented by the Hamiltonian, we introduce the associated interaction graph $G$ to depict overlaps of operators contained in $H$~\cite{haah2024learning,wild2023classical, wu2024efficient,zhang2024dequantized}. Specifically, given the Hamiltonian terms $\{h_X\}_{X\subset S}$, the interaction graph $G$ is a simple graph with vertex set $\{h_X\}_{X\subset S}$. An edge exists between $h_X$ and $h_{X^{\prime}}$ if $X\cap X^{\prime}\neq\emptyset$, and we denote the degree $\mathfrak{d}(h_X)$ of a vertex $h_X$, which is the number of edges incident to it. The maximum degree among all vertexes within the interaction graph $G$ is denoted by $\mathfrak{d}=\max_{h_X\in H}\left\{\mathfrak{d}(h_X)\right\}$.

For a set of $K$ unknown Hamiltonians $\{H^{(1)}, H^{(2)}, \dots, H^{(K)}\}$ defined by Eq.~\eqref{Eq:Hamiltonian}, and an evolution time series $\vec{t} = \{t_1, \dots, t_K\}$ with $\abs{t_k} = \mathcal{O}(1)$, we represent the corresponding Hamiltonian dynamics as $
U(\vec{t}) = \prod_{k=1}^{K} e^{-iH^{(k)}t_k}
$. Now, assuming that we have access only to the quantum dynamics $U(\vec{t})$, we define the problem of Hamiltonian dynamics learning as follows.

\begin{problem}
    [Hamiltonian Dynamics Learning]
\label{task1}
 Given access to $U(\vec{t})$, output an $n$-qubit channel $\mathcal{V}$ such that $\left\|\mathcal{V}-\mathcal{U}(\vec{t})\right\|_{\diamond}\leq\epsilon$ with high probability, where the unitary channel $\mathcal{U}(\vec{t})=U(\vec{t})(\cdot)U^{\dagger}(\vec{t})$.
\end{problem}

\noindent The diamond distance between quantum channels $\mathcal{V}$ and $\mathcal{U}$ is quantifies by $ \left\|\mathcal{V}-\mathcal{U}\right\|_{\diamond}=\max_{\sigma}\|(\mathcal{V}\otimes I)(\sigma)-(\mathcal{U}\otimes I)(\sigma)\|_1$, where $\sigma$ denotes all density matrices of $2n$ qubits.

\vspace{0.2cm}
\noindent\textbf{Quantum Learning Algorithm} --- 
Now, we introduce our Hamiltonian dynamics learning algorithm. Firstly, the algorithm exploits the operator identity established in Refs.~\cite{arrighi2011unitarity,huang2024learning}:
\begin{align}
U(\vec{t})\otimes U^{\dagger}(\vec{t})=S\prod\limits_{i=1}^n\left[U^{\dagger}(\vec{t})S_iU(\vec{t})\right],
\label{Eq:Identity}
\end{align}
where $S$ represents the $2n$-qubit SWAP operator and $S_i$ represents a $2$-qubit operator acting on qubit pair $(i,n+i)$. The basic idea behind this identity is that the involved global swap gates can induce a twist in the topological arrangement of the quantum circuit, thereby leaving the circuit itself unchanged. Then, by inserting $U(\vec{t})U^{\dagger}(\vec{t})$ after each $S_i$, the final identity is obtained. A visualization explanation on the proof of Eq.~\eqref{Eq:Identity} is provided in Fig.~\ref{fig:1} (a),(b). Crucially, the $2$-qubit SWAP operator can be decomposed as $S_i=\frac{1}{2}\sum_{O\in\{I,X,Y,Z\}}O_i\otimes O_{i+n}$. Substituting this decomposition into the identity yields 
\begin{align}\label{Eq:decompositin}
    U^{\dagger}(\vec{t})S_iU(\vec{t})=\frac{1}{2}\sum\limits_{O\in\{I,X,Y,Z\}}U^{\dagger}(\vec{t})O_iU(\vec{t})\otimes O_{i+n}
\end{align}
under the condition ${\rm supp}\left(U^{\dagger}(\vec{t})O_iU(\vec{t})\right)\cap {\rm supp}(O_{i+n})= \emptyset$. Here, the support of an operator $P$ represents the minimal qubit set such that $P=Q_{{\rm supp}(P)}\otimes I_{n\setminus{\rm supp}(P)}$ for some operator $Q$. When $U(\vec{t})$ is given by a $D$-dimensional Hamiltonian dynamics with constant evolution time, it can be approximated by quantum circuits with depth $\mathcal{O}\left(t{\rm poly}\log(nt)\right)$~\cite{haah2021quantum}. Combining this with the lightcone argument for information propagation in a $D$-dimensional lattice, we establish the support upper bound $\abs{{\rm supp}\left(U^{\dagger}(\vec{t})O_iU(\vec{t})\right)}\leq \mathcal{O}\left(t^D{\rm poly}\log(nt)\right)$, which implies the relationship ${\rm supp}\left(U^{\dagger}(\vec{t})O_iU(\vec{t})\right)\cap {\rm supp}(O_{i+n})= \emptyset$ when $t=\mathcal{O}(1)$. Therefore, learning the dynamics $U(\vec{t})$ can be effectively reduced to learning the set of operators $\left\{U_{O_i}(\vec{t})=U^{\dagger}(\vec{t})O_iU(\vec{t})\right\}_{i=1}^n$, which can subsequently be combined to reconstruct the full dynamics. 

To learn the Hermitian operator $U_{O_i}(\vec{t})$, one approach is to decompose $U_{O_i}(\vec{t})$ into a linear combinations of Pauli operators $\sum_{Q(O_i)}\alpha_{Q(O_i)}Q(O_i)$, and subsequently learn the real-valued coefficients $\alpha_{Q(O_i)}$. However, when $\abs{{\rm supp}(U_{O_i}(\vec{t}))}={\rm poly}\log n$, 
the direct decomposition may lead to some Pauli operators $Q(O_i)$ whose support size scales as ${\rm poly}\log n$. This further results in quasi-polynomial sample complexity and classical post-processing running time when using the brute-force computation. How to reduce the complexities to polynomial has been left as an open problem in Ref.~\cite{huang2024learning}. 
Here, we resolve this problem. We propose that \emph{learning coefficients of Pauli operators $Q(O_i)$ acting on $\mathcal{O}(\log n)$ qubits is sufficient to reconstruct the operator $U_{O_i}(\vec{t})$}. This statement is rigorously supported by the following lemma.
\begin{lemma}[Informal]
\label{lemma1}
Consider the Hamiltonian dynamics with evolution time $t=\max_k\{\abs{t_k}\}$, the operator $U_{O_i}(\vec{t})$ can be approximated by $V_{O_i}(\vec{t})=\sum_{Q(O_i)}\alpha_{Q(O_i)}Q(O_i)$ such that $\|\mathcal{U}_{O_i}(\vec{t})-\mathcal{V}_{O_i}(\vec{t})\|_{\diamond}\leq\epsilon^{\prime}\|O_i\|_{\infty}$. Here, $\max_{Q(O_i)}\abs{{\rm supp}(Q(O_i))}\leq \mathcal{O}(M(t))$ and $V_{O_i}(\vec{t})$ contains $L=\mathcal{O}((e\mathfrak{d})^{M(t)})$ Pauli operators $Q(O_i)$ with 
\begin{eqnarray}
M(t)=\left\{
    \begin{split}
        &\mathcal{O}\left(\frac{\log(1/\epsilon^{\prime})}{\log(t^*/t)}\right), ~\text{when}~t<t^*\\
        &\mathcal{O}\left(e^{\pi teK\mathfrak{d}}\log\left[\frac{e^{\pi teK\mathfrak{d}}}{\epsilon^{\prime}}\right]\right), ~\text{when}~t=\mathcal{O}(1)
    \end{split}
\right.
\end{eqnarray}
where the constant threshold $t^*=1/(2eK\mathfrak{d})$. Here, $\mathcal{U}_{O_i}(\vec{t})$, $\mathcal{V}_{O_i}(\vec{t})$ are channel representations of $U_{O_i}(\vec{t})$ and $V_{O_i}(\vec{t})$, respectively.
\end{lemma}
\noindent We demonstrate the cluster expansion induced Pauli decomposition as Fig.~\ref{fig:1}~(d), and leave the proof details to Appendixes~\ref{app:cluster}-\ref{App:E}. The above result demonstrates that short-time Hamiltonian dynamics restrict the spread of local information to a small region. To approximate the Hamiltonian dynamics $U(\vec{t})$ with an additive error of $\epsilon$ using Eq.~\eqref{Eq:Identity}, it is necessary to set $\epsilon^{\prime} = \mathcal{O}(\epsilon / n)$. This requirement leads to $M(t) = \mathcal{O}(\log(n/\epsilon))$, indicating that the support size of $Q(O_i)$ is independent of the Hamiltonian dimension $D$. Noting that the transition time $t=\max_k\{\abs{t_k}\}<1/(2eK\mathfrak{d})$ is taken by the analytical region of the cluster expansion function. When the evolution time is extended to the case $t=\mathcal{O}(1)$ via using the analytic continuum method, it may witness a phase transition in terms of the support size.

Given the above result, we outline the way to efficiently learn the operators $V_{O_i}(\vec{t})=\sum_{\abs{Q(O_i)}\leq M(t)}\alpha_{Q(O_i)}Q(O_i)$ for $i\in[n]$. The learning algorithm starts from preparing $N$ random tensor product states $\mathcal{D}(N)=\{|\psi_l\rangle=\otimes_{i=1}^n|\psi_{l,i}\rangle\}_{l=1}^N$, with the single qubit stabilizer state $|\psi_{l,i}\rangle\in\{|0\rangle,|1\rangle,|+\rangle,|-\rangle,|i+\rangle,|i-\rangle\}$. Then, the unknown quantum dynamics $U(\vec{t})$ (defined in Problems~\ref{task1}) is applied to the input quantum states $|\psi_l\rangle$ (demonstrated by Fig.~\ref{fig:1}~(c)), accompanied by the single-qubit Pauli measurement, resulting in output states $|\phi_l\rangle=\otimes_{i=1}^n|\phi_{l,i}\rangle$, where $|\phi_{l,i}\rangle$ also represents the single-qubit stabilizer state. For $O_i\in\{X_i,Y_i,Z_i\}$, quantum states $|\phi_l\rangle$ can be used to compute the random variable 
\begin{align}
    u_l(O_i)={\rm Tr}\left[\left(\otimes_{q=1}^n\left(3|\phi_{l,q}\rangle\langle\phi_{l,q}|-I\right)\right)O_i\right]=3\langle\phi_{l,i}|O_i|\phi_{l,i}\rangle
\end{align}
which satisfies $\mathbb{E}[u_l(O_i)]=\langle\psi_l|U_{O_i}(\vec{t})|\psi_l\rangle$. Here, the expectation $\mathbb{E}[\cdot]$ is defined over the single-qubit Pauli measurement to the quantum state $U(\vec{t})|\psi_l\rangle$. We then enumerate all Pauli operators $Q(O_i)$ such that ${\rm supp}(Q(O_i))\cap {\rm supp}(O_i)\neq\emptyset$ meanwhile $\abs{{\rm supp}(Q(O_i))}\leq \mathcal{O}(M(t))$, which composes the operator $V_{O_i}(\vec{t})=\sum_{Q(O_i)}\alpha_{Q(O_i)}Q(O_i)$. According to Lemma~\ref{lemma1}, it is shown that $V_{O_i}(\vec{t})$ only contains ${\rm poly}(n)$ valid Pauli operators $Q(O_i)$. Finally, all coefficients can be learned by $\alpha_{Q(O_i)}=\frac{3^{\abs{Q(O_i)}}}{N}\sum_{q=1}^Nu_q(O_i)\langle\psi_q|Q(O_i)|\psi_q\rangle$ which give rise to $V_{O_i}(\vec{t})$. Sew all learned operators $V_{O_i}(\vec{t})$ for $i\in[n]$ and $O\in\{X,Y,Z\}$, we finally obtain a quantum channel representation
\begin{align}
    \mathcal{V}(\rho)={\rm Tr}_{>n}\left[V\left(\rho\otimes I_n/2^n\right)V^{\dagger}\right],
    \label{Eq:learnedChannel}
\end{align}
with the operator
\begin{align}
    V=S\prod\limits_{i=1}^n\left[\frac{1}{2}\sum\limits_{O\in\{I,X,Y,Z\}}V_{O_i}(\vec{t})\otimes O_{i+n}\right],
    \label{Eq:learnedCircuit}
\end{align}
an approximation to the unknown quantum dynamics $U(\vec{t})(\cdot)U^{\dagger}(\vec{t})$. We summarize the learning algorithm as Alg.~\ref{Algorithm}, and we leave more details to Appendix~\ref{App:E} and \ref{app:sew}.

\begin{algorithm}
\label{Algorithm}
\caption{Hamiltonian Dynamics Learning Algorithm}
\textbf{Input:} $N$ random tensor product states $\mathcal{D}(N)=\{|\psi_l\rangle=\otimes_{i=1}^n|\psi_{l,i}\rangle\}_{l=1}^N$, target Hamiltonian dynamics $U(\vec{t})$.\\
\textbf{Output:} Quantum channel $\mathcal{V}$ such that $\|\mathcal{V}-\mathcal{U}(\vec{t})\|_{\diamond}\leq\epsilon$.\\
\textbf{For} $l\in[N]$:\\
\quad Apply $U(\vec{t})$ to $|\psi_l\rangle$;\\
\quad Measure the output state by random Pauli basis, obtain $|\phi_l\rangle$;\\
\textbf{End For}\\
\textbf{For} $i\in[n]$, $O_i\in\{X_i,Y_i,Z_i\}$:\\
\quad Compute $u_l(O_i)=3\langle\phi_{l,i}|O_i|\phi_{l,i}\rangle$;\\
\quad Enumerate all Pauli operators $Q(O_i)$ s.t. ${\rm supp}(Q(O_i))\cap {\rm supp}(O_i)\neq\emptyset$, $\abs{{\rm supp}(Q)}\leq M(t)$ (promised by Lemma~\ref{lemma1});\\
\quad \textbf{For}: Valid $Q(O_i)$\\
\quad\quad Compute $\alpha_{Q(O_i)}=\frac{3^{\abs{Q(O_i)}}}{N}\sum_{l=1}^Nu_l(O_i)\langle\psi_l|Q(O_i)|\psi_l\rangle$\\
\quad \textbf{End For}\\
\textbf{End For}\\
\textbf{Output}: Use $V_{O_i}(\vec{t})=\sum_{Q(O_i)}\alpha_{Q(O_i)}Q(O_i)$ for $i\in[n]$ and $O_i\in\{X_i,Y_i,Z_i\}$ to construct Eqs.~\eqref{Eq:learnedChannel} and~\eqref{Eq:learnedCircuit}\\
{\textbf{End}}\\
\end{algorithm}

\vspace{0.2cm}
\noindent\textbf{Theoretical Guarantee} ---
Here, we show that the proposed quantum learning algorithm is highly efficient in both sample complexity and computational complexity.
\begin{theorem}
    Given an error $\epsilon$, failure probability $\delta$, an unknown quantum dynamics $U(\vec{t})$, there exists a learning algorithm that requires 
    \begin{align}
        N=\frac{n^2(4^{K\Lambda}3e\mathfrak{d})^{\mathcal{O}(M(t))}\log(1/\delta)}{\epsilon^2}
    \end{align}
  quantum measurements and $\mathcal{O}\left(N\left(4^{K\Lambda }e\mathfrak{d}\right)^{M(t)}\right)$ classical post-processing time to reconstruct a quantum channel $\mathcal{V}$ such that $\|\mathcal{V}-\mathcal{U}\|_{\diamond}\leq\epsilon$, with the succuss probability $\geq 1-\delta$ and $M(t)$ is given by Lemma~\ref{lemma1}.
\end{theorem}
\noindent We leave proof details to the Appendixes~\ref{shorttime} and~\ref{longtime}. It states that constant parameters $K$, $\Lambda$ and $\mathfrak{d}$ may imply polynomial sample complexity $N$ and the classical computational complexity in terms of the number of qubits. 

This result partially solves the Hamiltonian dynamics learning problem as they rely on an implicit decomposed form of Eq.~\eqref{Eq:decompositin}. A further nontrivial challenge is to reconstruct the unitary operation based on this decomposition explicitly. Here, we demonstrate that such a construction is feasible with polynomial quantum and classical complexities.

It is crucial to note that the learned local evolution $V_{O_i}(\vec{t})$ acts non-trivially on a $D$-dimensional region involving $\mathcal{O}(\log^D(n))$ qubits, which fundamentally limits efficient compilation via brute-force computation. Nevertheless, we observe that $V_i(\vec{t})=\frac{1}{2}\sum_{O_i\in\{X,Y,Z\}}V_{O_i}(\vec{t})\otimes O_{i+n}$ provides an approximation to the Hermitian operator $U^{\dagger}(\vec{t})S_iU(\vec{t})$, with eigenvalues $\{-1,+1\}$. This key insight enables simulation of $V_i(\vec{t})$ through Hamiltonian dynamics $e^{\frac{-i\pi}{2}(V_i(\vec{t})-I)}$, which approximates $U^{\dagger}(\vec{t})S_iU(\vec{t})$ with $\mathcal{O}(\epsilon)$ additive error in diamond norm. (Technical details are deferred to Appendix~\ref{App:compile}.) Consequently, the quantum channel $\mathcal{V}$ can be compiled as $\mathcal{V}(\rho)={\rm Tr}_{>n}\left[V^{\prime}(\rho\otimes I_n/2^n)(V^{\prime})^{\dagger}\right]$, where 
\begin{align}
    V^{\prime}=S\prod\limits_{i=1}^ne^{\frac{-i\pi}{2}(V_i(\vec{t})-I)},
\end{align}
and its circuit depth complexity formally characterized by the following theorem. We summarize the result as follows. 

\begin{theorem}
\label{them:compile}
    Given the Hamiltonian dynamics $U(\vec{t})=\prod_{k=1}^Ke^{-iH^{(k)}t_k}$. Suppose $t=\max_k\{\abs{t_k}\}$, the quantum channel $\mathcal{U}=U^{\dagger}(\vec{t})(\cdot)U(\vec{t})$ can be approximated by a unitary channel $\mathcal{V}$ of depth
    \begin{align}
        \mathcal{O}\left(\log^D(n) \left[4^{K\Lambda}e\mathfrak{d}\right]^{M(t)}/\epsilon^{1/p} \right)
    \end{align}
    with approximation error $\leq\mathcal{O}(\epsilon)$ in diamond norm. Here, parameters $K$, $\Lambda$, $\mathfrak{d}$ and $p$ are constants.
\end{theorem}
\noindent The above result can be derived via the gate counting argument. Specifically, each local Hamiltonian dynamics $e^{-i\pi/2(V_i(\vec{t})-I)}$ can be efficiently simulated by using the $p$-th order Trotter-Suzuki method~\cite{childs2021theory}, with the quantum circuit depth $\mathcal{O}(\left(4^{K\Lambda}e\mathfrak{d}\right)^{M(t)}/\epsilon^{1/p})$. Finally, we note that $\mathcal{O}(n/\log^D(n))$ local evolutions $e^{-i\pi/2(V_i(\vec{t})-I)}$ can be implemented simultaneously, as a result, the circuit depth of $V$ follows the result given in Theorem~\ref{them:compile}.

\vspace{0.2cm}
\noindent\textbf{Applications} --- 
The Hamiltonian dynamics learning algorithm has wide applications in quantum machine learning, quantum computation verification, and realistic device benchmarking.

\emph{Quantum machine learning}: The quantum neural network is one of the representative models in the field of near-term quantum machine learning. Various applications of learning quantum circuits and quantum dynamics have been explored, ranging from compressing quantum circuits for implementing a unitary~\cite{caro2022generalization, cincio2018learning,khatri2019quantum, sharma2020noise, jones2022robust}, speeding up quantum dynamics~\cite{cirstoiu2020variational,yao2021adaptive, gibbs2024dynamical,caro2023out}, to learning generative models for sampling from predicted distributions~\cite{lloyd2018quantum,dallaire2018quantum,hu2019quantum,benedetti2019generative, gao2022enhancing}. Although the optimization landscape, induced by short-time Hamiltonian variational (HV) ansatz, may not suffer from the barren plateaus phenomenon~\cite{park2024hamiltonian}, how to efficiently train general HV ansatz parameters is still an open problem~\cite{song2023trainability}. Here, we show that the HV ansatz based quantum machine learning model may be efficiently trained.

\begin{problem}
    [Quantum variational classifier]
\label{task3}
Considering $K$ Hamiltonians $\{H^{(1)},H^{(2)},\cdots, H^{(K)}\}$ defined by Eq.~\eqref{Eq:Hamiltonian} and variational parameters $\vec{\bm \theta}=\{\theta_1,\cdots,\theta_K\}$ with $\abs{\theta_k}\leq\mathcal{O}(1)$, denote the HV ansatz $U(\vec{\bm\theta})=\prod_{k=1}^Ke^{-iH^{(k)}\theta_k}$. Suppose a classical data set $\mathcal{D}=\{(x_i,y_i)\}$, and the quantum classifier tries to learn the optimal $\vec{\bm\theta}$ such that
\begin{align}
    {\rm{Loss}}(\vec{\bm\theta})=\mathbb{E}_{(x_i,y_i)\sim\mathcal{D}}\abs{\langle \phi(x_i)|U^{\dagger}(\vec{\bm \theta})OU(\vec{\bm \theta})|\phi(x_i)\rangle-y_i}\leq\epsilon,
\end{align}
where $|\phi(x_i)\rangle$ represents a quantum feature map of $x_i$, $O$ represents a local observable, and training error $\epsilon$.
\end{problem}
Here, we argue that Lemmas~\ref{lemma1} implies an efficient learning algorithm for the above problem. 
\begin{corollary}
Given problem~\ref{task3}, suppose the optimal parameter $\vec{\bm\theta}^*=\arg\min_{\vec{\bm\theta}}{\rm{Loss}}(\vec{\bm\theta})$, then there exists a quantum-classical algorithm that can output a model $f(\cdot)$ such that 
\begin{align}
\mathbb{E}_{x_i\sim\mathcal{D}}\abs{f\left(|\phi(x_i)\rangle\langle\phi(x_i)|\right)-\langle \phi(x_i)|U^{\dagger}(\vec{\bm \theta^*})OU(\vec{\bm \theta^*})|\phi(x_i)\rangle}\leq\epsilon
\end{align}
 with the running time $\abs{\mathcal{D}}\left[4^{K\Lambda}e\mathfrak{d}\right]^{\mathcal{O}(\log(n/\epsilon))}$.
\end{corollary}

\emph{Quantum Computation Verification}: An efficient approach to verify large-scale quantum computation is significant for testing the stability and reliability of realistic quantum computers. However, as noted by Gottesman~\cite{barnum2002authentication, Aaronsonprize}: ``If a quantum computer can efficiently solve a problem, can it also efficiently convince an observer that the solution
is correct?", this problem is significantly challenging. From the computational complexity theory, it is related to the relationship between ${\rm BQP}$ and ${\rm IP}$, where whether ${\rm BQP}\subseteq{\rm IP}$ is still an open problem. Our results (lemma~\ref{lemma1}) indicate the possibility of efficient quantum dynamics verification when the evolution time is constant.

\begin{corollary}
    Given the target Hamiltonian dynamics $U(\vec{t})$, any local observable s.t. $\|O\|\leq 1$ and a classical simulable state $|\phi\rangle$ with $R$ configurations, the output of Alg.~\ref{Algorithm} can be used to verify and predict the quantum mean value $\langle\phi|U^{\dagger}(\vec{t})OU(\vec{t})|\phi\rangle$ within $\epsilon$ additive error, with the running time $R^2\left[4^{K\Lambda}e\mathfrak{d}\right]^{\mathcal{O}(\log(n/\epsilon))}$.
\end{corollary}

Particularly, when the unknown Hamiltonian dynamics exhibits specific geometrical structures, the above result can be extended to verify the mean value of global observables.
\begin{corollary}
If the target Hamiltonian dynamics $U(\vec{t})$ is confined to a 2D architecture, then the output of Alg.~\ref{Algorithm} can be used to verify and predict the quantum mean value $\langle\phi|U^{\dagger}(\vec{t})(O_1\otimes\cdots\otimes O_n)U(\vec{t})|\phi\rangle$ within $\epsilon$ additive error, with observable $\|O_i\|\leq1$ and classical state $|\phi\rangle$ with $R$ configurations. The classical running time is $\mathcal{O}(R^2n^{\log(n/\epsilon)})$. 
\label{coro:2D}
\end{corollary}

\emph{Noisy Quantum Device Benchmarking:} Realistic analog and digital quantum computers generally exhibit a certain level of noise, leading the quantum circuit to behave as a more complicated channel. Here, we treat $U(\vec{t})$ as a constant-layer analog quantum circuit and study the performance of Alg.~\ref{Algorithm} in the scenario where each layer is affected by a $\gamma$-strength depolarizing channel $\mathcal{N}=\otimes_{i=1}^n\mathcal{N}_i$ with $\mathcal{N}_i(\cdot)=(1-\gamma)(\cdot)+\gamma\frac{I}{2}{\rm Tr}(\cdot)$. 
\begin{corollary}
  Given the noisy quantum analog circuit $\mathcal{U}_{\rm noisy}(\vec{t})=\mathcal{N}\circ \mathcal{U}_{K}\circ \mathcal{N}\circ \mathcal{U}_{K-1}\circ\cdots\circ\mathcal{N}\circ \mathcal{U}_{1}$, Alg.~\ref{Algorithm} may output a $2n$-qubit channel $\tilde{\mathcal{U}}$ such that
     \begin{eqnarray}
         \begin{split}
             \bigg|&{\rm Tr}\left[O\left(\mathcal{U}_{\rm noisy}(\vec{t})(|\psi\rangle\langle\psi|)\otimes \frac{I_n}{2^n}\right)\right]\\
             &-{\rm Tr}\left[O\tilde{\mathcal{U}}\left(|\psi\rangle\langle\psi|\otimes \frac{I_n}{2^n}\right)\right]\bigg|\leq \mathcal{O}(\gamma n^2\|O\|_{\infty})
         \end{split}
     \end{eqnarray}
    for any quantum state $|\psi\rangle$ and $n$-qubit observable $O$.
\end{corollary}

\noindent Our result indicates the robustness of Alg.~\ref{Algorithm} when benchmarking weakly noisy quantum circuits. In particular, 
for weak noise with strength $\gamma=\epsilon/n^{2}$, the learning algorithm may approximate the output state produced by the noisy quantum circuit, in terms of predicting the quantum mean value induced by local observables. For strong noise with non-negligible $\gamma n^2 = O(1)$, we left it as an open problem to benchmark a general noisy channel. 

\vspace{0.2cm}

\noindent\textbf{Discussion and Outlook} ---
This work resolves the fundamental challenge of short-time Hamiltonian dynamics learning by establishing an efficient framework to reconstruct constant-time quantum dynamics generated by local Hamiltonians. Our work resolves an open problem in Huang et al.~\cite{huang2024learning} by extending the methodology originally designed for constant-depth quantum circuits to time-evolved Hamiltonian dynamics, and bridges a critical gap in quantum process characterization. These advances provide rigorous theoretical foundations for analyzing quantum machine learning models, verifying computational outcomes, and benchmarking noisy quantum processes in near-term quantum devices.

Our work opens several avenues for future exploration. First, although the current protocol focuses on unitary process tomography, extending it to general completely positive trace-preserving map tomography would significantly broaden its applicability to noisy quantum processes, enabling the characterization of decoherence and error channels. Second, it would be valuable to relax the current constant-time evolution constraint to allow for logarithmic-time evolution. Finally, the polynomial dependence of our algorithm's complexity on $n$ and $1/\epsilon^2$ raises a natural open question: can these scalings be further optimized? Addressing these challenges could advance our understanding of quantum dynamics learning and enhance its practical utility in the context of near-term quantum computing.

\vspace{8pt}


\section*{Acknowledgments}
\noindent The authors wish to express gratitude to Jeongwan Haah, Yu Tong and Anurag Anshu for valuable discussions. Y.~Wu and C.~Wang acknowledge the support from the Natural Science Foundation of China Grant (No.~62371050). X.~Yuan acknowledges the support from Innovation Program for Quantum Science and Technology Grant (No.~2023ZD0300200), the National Natural Science Foundation of China Grant (No.~12175003 and No.~12361161602), and NSAF Grant (No.~U2330201).

\bibliography{main.bbl}
\clearpage

\widetext
\section*{Supplementary Information}
\appendix

\section{Comparison to Related works}
\subsection{Hamiltonian Learning}
Hamiltonian learning from real-time evolution has wide applications in the fields of quantum metrology and quantum sensing. The problem can be informally defined as follows: consider an \( n \)-qubit quantum system with Hamiltonian $H = \sum_{X \in S} \lambda_X h_X$, where \( \mathrm{supp}(h_X) = \mathcal{O}(1) \) and \( |\lambda_X| \leq 1 \). The goal is to determine the coefficients \( \lambda_X \) for \( X \in S \), given the ability to evolve quantum states under the unitary \( e^{-iHt} \) for any time \( t > 0 \). Hamiltonian learning algorithms focus on the total evolution time \( t_{\mathrm{total}} \) and time resolution \( t_{\mathrm{min}} \), where \( t_{\mathrm{total}} \) generally determines the total runtime complexity, and a short time resolution \( t_{\mathrm{min}} \) imposes stricter requirements on the precise control of quantum devices.

Hamiltonian learning algorithms can be broadly classified into two categories based on whether the operators \( \{h_X\} \) are known in prior. For example, Refs.~\cite{huang2023learning, haah2024learning} assume prior knowledge of the Hamiltonian structure, whereas recent works~\cite{zhao2024learning, bakshi24structure} have developed algorithms that do not require knowledge of \( \{h_X\} \). However, these methods still rely on the ability to control \( e^{-iHt} \) for continuous time \( t > t_{\mathrm{min}} \). In contrast, the \emph{Hamiltonian dynamics learning} problem attemptes to remove the assumption of access to \( e^{-iHt} \) for any \( t > t_{\mathrm{min}} \), and neither the time \( t \) nor the internal structure of the Hamiltonian is assumed to be known. 

\subsection{Quantum Process Tomography}
The gold standard for the full characterization of quantum circuits is the quantum process tomography, a procedure that reconstructs an unknown quantum process from quantum data. A direct approach relies on a complete set of measurement operator, which naturally introduces exponentially large sample and post-processing complexity. To address this challenge, various heuristic approaches have been developed, leveraging parameterized quantum ansätze~\cite{xue2022variational,xue2023variational,wulearning}, Bayesian inference~\cite{granade2012robust, wiebe2014hamiltonian,stenberg2014efficient}, neural network models~\cite{xin2019local, melko2019restricted}, and tensor networks~\cite{lanyon2017efficient,torlai2023quantum}, while it is generally hard to study computational complexity and convergence rate for variational methods, making their applicability and scalability uncertain.  

\subsection{Quantum Circuit Learning}
In the context of supervised learning theory, the objective of a learning algorithm can be abstracted as follows: given a training set \(\mathcal{T} = \{(x_i, y_i)\}_{i=1}^N\), where \(x_i\) are the data sampled from the data space \(\mathcal{X}\) with probability \(D(x_i)\), and \(y_i\) are the corresponding labels, learn a map \(f\) such that
\[
\mathbb{E}_{x \sim D} \left[ |f(x) - y| \right] \leq \epsilon,
\]
where \(\epsilon\) is a small learning error~\cite{mohri2018foundations}. Note that the target map \(f\) can be equivalently represented by a classical circuit \(C_f\), leading to a broadly interesting question: What types of classical circuits \(C_f\) can be efficiently learned with polynomial sample complexity and polynomial running time? It is known that constant-depth classical circuits with bounded fan-in gates (\(\text{NC}^0\)) can be efficiently learned, while constant-depth classical circuits with unbounded fan-in gates (\(\text{AC}^0\)) require quasi-polynomial running time. Given that quantum circuits may generate distributions that are classically hard to simulate, it is natural to ask: Can a quantum circuit \(U\) be learned efficiently?

However, this challenge is fundamental in the quantum world. Over the past decade, variational quantum algorithms have been proposed to learn an approximation of a quantum circuit using a parameterized quantum ansatz~\cite{mitarai2018quantum, schuld2019quantum, havlivcek2019supervised, huang2021power, abbas2021power,wu2023orbital}. Although relationships between sample complexity, the number of variational quantum gates, and generalization error have been theoretically studied~\cite{abbas2021power, caro2022generalization, caro2023out}---where the required sample complexity \(N\) depends polynomially on the number of variational parameters and \(1/\epsilon\)---the computational complexity remains unclear. This uncertainty fundamentally arises from the optimization landscape induced by the parameterized quantum circuit, which may have exponentially small local minima in shallow-depth circuits~\cite{anschuetz2022quantum} and suffer from the barren plateau phenomenon in deeper circuits~\cite{mcclean2018barren, cerezo2021cost, ragone2024lie, fontana2024characterizing}. These phenomena significantly hinder the efficiency of variational approaches in learning quantum circuits, both in shallow and deep scenarios.

Recently, a breakthrough work has proposed an efficient learning algorithm for constant-depth quantum circuits~\cite{huang2024learning}. The key idea is to transform the learning of \(U\) into learning the local evolution \(U^\dagger O_i U\), where \(O_i \in \{X, Y, Z\}\) and \(i \in [n]\). When \(U\) represents a \(d = \mathcal{O}(1)\)-depth quantum circuit, the support size of the operator \(U^\dagger O_i U\) is constant, as information only spreads within the \textit{light cone}. As a result, it can be efficiently learned using random stabilizer input states and random measurement outputs. However, when the quantum circuit depth \(d \geq \Omega(\log n)\), the corresponding support size may grow to \(\mathcal{O}(n)\) (for all-to-all connection architectures) or \(\text{poly}(\log n)\) (for \(D\)-dimensional lattices), making the learning algorithm no longer efficient.

\section{Preliminary knowledge}
\begin{problem}
    [Quantum Dynamics Learning]
\label{problem1_app}
Consider $K$ local Hamiltonians $$\{H^{(1)},H^{(2)},\cdots, H^{(K)}\},$$ and evolution time series $\vec{t}=\{t_1,\cdots,t_K\}$ with $\abs{t_k}\leq\mathcal{O}(1)$. Denote $U(\vec{t})=\prod_{k=1}^Ke^{-iH^{(k)}t_k}$, and the target is to learn an $n$-qubit channel $\mathcal{V}$ such that
\begin{align}
    \left\|\mathcal{V}-\mathcal{U}(\vec{t})\right\|_{\diamond}\leq\epsilon
\end{align}
with high probability, where the channel $\mathcal{U}(\vec{t})=U(\vec{t})(\cdot)U^{\dagger}(\vec{t})$.
\end{problem}

\begin{definition}[Diamond norm]
The diamond norm is the trace norm of the output of a trivial extension of a linear map, maximized over all possible inputs with trace norm at most one. Specifically, suppose $\rho:H_d\mapsto H_d$ be a linear transformation, where $H_d$ be a $d$-dimensional Hilbert space, the diamond norm of $\rho$ is defined by
\begin{align}
    \|\rho\|_{\diamond}=\max\limits_{\|X\|_1\leq 1}\|(\rho\otimes I)X\|_1,
\end{align}
where identity matrix $I\in H_d$, $X\in H_{2d}$ and $\|\cdot\|_1$ represents the trace norm.
\end{definition}

\begin{definition}[Diamond distance]
Given $n$-qubit quantum channels $\mathcal{G}$ and $\mathcal{F}$, their diamond distance is defined as 
\begin{align}
    d_{\diamond}(\mathcal{G},\mathcal{F})=\|\mathcal{G}-\mathcal{F}\|_{\diamond}=\max\limits_{\rho}\|(\mathcal{G}\otimes I)(\rho)-(\mathcal{F}\otimes I)(\rho)\|_1,
\end{align}
where $\rho$ denotes all density matrices of $2n$ qubits.
\end{definition}

\begin{lemma}[Diamond distance for unitaries~\cite{haah2023query}]
For any two unitaries $U_1$ and $U_2$, we have
\begin{align}
    \min\limits_{\phi\in\mathbb{R}}\|e^{i\phi}U_1-U_1\|_{\infty}\leq\|\mathcal{U}_1-\mathcal{U}_2\|_{\diamond}\leq 2\min\limits_{\phi\in\mathbb{R}}\|e^{i\phi}U_1-U_2\|_{\infty}.
\end{align}
\end{lemma}

\begin{definition}[Operator norm]
    Given a matrix A, its operator norm is defined by 
    \begin{align}
        \|A\|_{\infty}=\max\limits_{|\psi\rangle}\abs{\langle\psi|A|\psi\rangle}.
    \end{align}
\end{definition}

It is shown that the operator-norm based learning metric suffices to provide an estimation to the quantum mean value problem. To check this fact, we suppose an approximation $\hat{U}$ satisfying $\|U-\hat{U}\|_{\infty}\leq\epsilon$. For any valid quantum state $|\psi\rangle$ and an observable $O$, we have
\begin{eqnarray}
    \begin{split}
        \abs{\langle\psi|U^{\dagger}OU|\psi\rangle-\langle\psi|\hat{U}^{\dagger}O\hat{U}|\psi\rangle}
        \leq \left\|\left(U^{\dagger}OU-\hat{U}^{\dagger}O\hat{U}\right)|\psi\rangle\right\|_{\infty}
        \leq  \|(U^{\dagger}-\hat{U}^{\dagger})OU+\hat{U}^{\dagger}O(U-\hat{U})\|_{\infty} 
        \leq 2\|O\|_{\infty}\epsilon,
    \end{split}
\end{eqnarray}
where the first inequality comes from the Cauchy-Schwarz inequality.

\section{Learning Algorithm Outline}
Our learning algorithm is fundamentally based on the identity
\begin{align}
    U(\vec{t})\otimes U^{\dagger}(\vec{t})=S\prod\limits_{i=1}^n\left[U^{\dagger}(\vec{t})S_iU(\vec{t})\right],
\end{align}
where $S$ represents a $2n$-qubit SWAP operator, and $S_i$ represents a $2$-qubit operator on the qubit pair $(i,n+i)$. Noting that the $2$-qubit SWAP operator $S_i=\frac{1}{2}\sum_{O\in\{I,X,Y,Z\}}O_i\otimes O_{i+n}$, and this yields 
\begin{align}
    U^{\dagger}(\vec{t})S_iU(\vec{t})=\frac{1}{2}\sum\limits_{O\in\{I,X,Y,Z\}}U^{\dagger}(\vec{t})O_iU(\vec{t})\otimes O_{i+n}
\end{align}
when the relationship 
\begin{align}
    {\rm supp}\left(U^{\dagger}(\vec{t})O_iU(\vec{t})\right)\cap {\rm supp}(O_{i+n})\neq \emptyset
\end{align}
holds. From the above Heisenberg picture, it is shown that learning the quantum dynamics $U(\vec{t})$ can be transformed into learning operators $$U_i(\vec{t})=U^{\dagger}(\vec{t})O_iU(\vec{t})$$ for qubit index $i\in[n]$. 

To learn the operator $U_i(\vec{t})$, Ref.~\cite{huang2024learning} decomposed it into linear combinations of Pauli operators, where each operator living in its lightcone. When the quantum circuit is constant depth, one can easily enumerate all possible Pauli operators. However, $\Omega(\log n)$-depth $D$-dimensional ($D\geq 2$) quantum circuit may lead to a ${\rm poly}\log n$-sized lightcone, and the number of resulting Pauli operators grows quasi-polynomially with the number of qubits. To overcome the large sample complexity due to quasi-polynomial number of Pauli operators, we utilize the cluster expansion method to approximate $U^{\dagger}(\vec{t})O_iU(\vec{t})$ by $V_{O_i}(\vec{t})$ in a small size lightcone, then utilize the randomized measurement dataset to learn the expression of $V_i(\vec{t})$, an approximation to $U^{\dagger}(\vec{t})S_iU(\vec{t})$. Finally, the learned channel can be expressed by
\begin{align}
   \mathcal{V}(\cdot)={\rm Tr}_{>n}\left[\left(S\prod\limits_{i=1}^nV_i(\vec{t})\right)(\cdot\otimes I_n2^n)\left(S\prod\limits_{i=1}^nV_i(\vec{t})\right)^{\dagger}\right].
\end{align}
In the following sections, we detail how to approximate $U^{\dagger}(\vec{t})O_iU(\vec{t})$ by using the cluster expansion method given by Refs.~\cite{wild2023classical,wu2024efficient}.

\section{Approximate local evolution $U^{\dagger}
(\vec{t})O_iU(\vec{t})$}
\label{app:cluster}
\subsection{Cluster induced by local Hamiltonian}
\begin{definition}[Cluster induced by Hamiltonian]
    Given a $D$-dimensional Hamiltonian $$H=\sum_{X\in S}\lambda_Xh_X,$$
    where the coefficient $\abs{\lambda_X}\leq 1$, each term satisfies $\|h_X\|\leq1$ and $S=\{X\}$ represents the qubit set performed by $h_X$. A cluster $\bm{W}$ is defined as a nonempty multi-set of subsystems from $S$, where multi-sets allow an element appearing multiple times. The set of all clusters $\bm{W}$ with size $m$ is denoted by $\mathcal{C}_m$ and the set of all clusters is represented by $\mathcal{C}=\cup_{m\geq 1}\mathcal{C}_m$.
\end{definition}

For example, if the Hamiltonian $H=X_0X_1+Y_0Y_1$, then some possible candidates for $\bm{W}$ would be $\{X_0X_1\}$, $\{Y_0Y_1\}$, $\{X_0X_1,X_0X_1\},\cdots$. We call the number of times a subsystem $X$ appears in a cluster $\bm{W}$ the multiplicity $\mu_{\bm{W}}(X)$, otherwise we assign $\mu_{\bm{W}}(X)=0$. Traversing all subsets $X\in S$ may determine the size of $\bm{W}$, that is $\abs{\bm{W}}=\sum_{X\in S}\mu_{\bm{W}}(X)$. In the provided example, when $\bm{W}=\{X_0X_1,X_0X_1\}$, we have $\mu_{\bm{W}}(X_0X_1)=2$, $\mu_{\bm{W}}(Y_0Y_1)=0$ and $\abs{\bm{W}}=2$.

\begin{definition}[Interaction Graph]
    We associate with every cluster $\bm{W}$ a simple graph $G_{\bm{W}}$ which is also termed as the cluster graph. The vertices of $G_{\bm{W}}$ correspond to the subsystems in $\bm{W}$, with repeated subsystems also appearing as repeated vertices. Two distinct vertices $X$ and $Y$ are connected by an edge if and only if the respective subsystems overlap, that is ${\rm supp}(h_X)\cap{\rm supp}(h_Y)\neq\emptyset$.
\end{definition}

Suppose the cluster $\bm{W}=\{X_0X_1,X_0X_1\}$, then its corresponding interaction graph $G_{\bm{W}}$ has two vertices $v_1,v_2$, related to $X_0X_1$ and $X_0X_1$, respectively, and $v_1$ connects to $v_2$ since ${\rm supp}(X_0X_1)\cap {\rm supp}(X_0X_1)\neq \emptyset$. We say a cluster $\bm{W}$ is connected if and only if $G_{\bm{W}}$ is connected. We use the notation $\mathcal{G}_m$ to represent all connected clusters of size $m$ and $\mathcal{G}=\cup_{m\geq 1}\mathcal{G}_m$ for the set of all connected clusters.

\begin{definition}[Super-Interaction Graph]
\label{Def:superinteractiongraph}
    Suppose we have $K$ clusters $\bm{W}_1,\bm{W}_2,\cdots, \bm{W}_K$, we define the super-interaction graph $G^K_{\bm{W}_1,\cdots,\bm{W}_K}$ composed by interaction graphs $G_{\bm{W}_1},G_{\bm{W}_2},\cdots G_{\bm{W}_K}$, where vertices $\{h_X\}_{X\in S_1\cup S_2\cdots\cup S_K}$ inherit from $G_{\bm{W}_1},G_{\bm{W}_2},\cdots G_{\bm{W}_K}$  and vertices $h_X$ and $h_Y$ are connected if ${\rm supp}(h_X)\cap{\rm supp}(h_Y)\neq\emptyset$. 
\end{definition}

In our paper, the super-interaction graph is generally induced by a Hamiltonian series, say $\{H^{(1)},\cdots,H^{(K)}\}$. From the above definition, we know that the super-interaction graph $G^K_{\bm{W}_1,\cdots,\bm{W}_K}$ contains $\sum_{k=1}^K\abs{G_{\bm{W}_k}}$ vertices. 

\begin{definition}[Connected Super-Interaction Graph]
\label{Def:consuperinteractiongraph}
The super-interaction-graph $G^K_{\bm{W}_1,\cdots,\bm{W}_K}$ is connected if and only if 
the super cluster $\bm{W}=(\bm{W}_1,\bm{W}_2,\cdots, \bm{W}_K)$ is connected. All $m$-sized connected super-interaction graphs are denoted by $\mathcal{G}_m^K$, with $m=\sum_{k=1}^K\abs{\bm{W}_k}$.
\end{definition}

Specifically, we denote $\mathcal{G}_m^{K,O_i}$ as the set of all $m$-sized connected super-interaction graphs which connects to $O_i$.

\subsection{Cluster Expansion}
\label{sec:cluster expansion}
We first consider a simple case, that is the cluster expansion of the single-step Hamiltonian dynamics $e^{iHt}O_ie^{-iHt}$~\cite{wild2023classical}. For any cluster $\bm{W}\in\mathcal{C}_m$, we can write $\bm{W}=(X_1,\cdots, X_m)$. This notation helps us to write the function $e^{iHt}O_ie^{-iHt}$ as the multivariate Taylor-series expansion by using the cluster expansion method. Here, we fix the parameter $O_i$, but considering $\{t,\lambda_X\}$ as variables. As a result, we have
\begin{eqnarray}
\begin{split}
    e^{iHt}O_ie^{-iHt}=&\sum\limits_{m=0}^{+\infty}\frac{t^m}{m!}\left(\frac{\partial^m[e^{iHt}O_ie^{-iHt}]}{\partial t^m}\right)_{t=0}
    \label{Eq:taylor}
\end{split}
\end{eqnarray}
Recall that the Hamiltonian $H=\sum_{X\in S}\lambda_Xh_X$, then we assign $z_X=-it\lambda_X$. This results in 
\begin{align}
    \frac{\partial[e^{iHt}O_ie^{-iHt}]}{\partial t}=\sum_{X\in S}\frac{\partial[e^{iHt}O_ie^{-iHt}]}{\partial z_X}\frac{\partial z_X}{\partial t}=\sum_{X\in S}(-i)\lambda_X\frac{\partial[e^{iHt}O_ie^{-iHt}]}{\partial z_X}.
\end{align}
Taking above derivative function into Eq.~\eqref{Eq:taylor}, we have
\begin{eqnarray}
    \begin{split}
         e^{iHt}O_ie^{-iHt}=&\sum\limits_{m=0}^{+\infty}\frac{(-it)^m}{m!}\sum\limits_{X_1,\cdots,X_m}\lambda_{X_1}\cdots\lambda_{X_m}\left(\frac{\partial^m[e^{iHt}O_ie^{-iHt}]}{\partial z_{X_1}\cdots\partial z_{X_m}}\right)_{z=(0,\cdots, 0)}\\
         =&\sum\limits_{m=0}^{+\infty}(-it)^m\sum\limits_{\bm{W}\in\mathcal{C}_m,V=(X_1,\cdots,X_m)}\frac{\bm\lambda^{\bm{W}}}{\bm{W}!}\left(\frac{\partial^m[e^{iHt}O_ie^{-iHt}]}{\partial z_{X_1}\cdots\partial z_{X_m}}\right)_{z=(0,\cdots, 0)}
    \end{split}
\end{eqnarray}
where $\bm\lambda^{\bm{W}}=\prod_{X\in S}\lambda_X^{\mu_{\bm{W}}(X)}$ and $\bm{W}!=\prod_{X\in S}\mu_{\bm{W}}(X)!$. Finally, we utilize the BCH expansion to compute
\begin{eqnarray}
    \begin{split}
        \left(\frac{\partial^m[e^{iHt}O_ie^{-iHt}]}{\partial z_{X_1}\cdots\partial z_{X_m}}\right)_{z=(0,\cdots, 0)}=&\frac{\partial^m}{\partial z_{X_1}\cdots\partial z_{X_m}}\sum\limits_{j=0}^{\infty}\frac{(-it)^j}{j!}[H,O_i]_j\big|_{z=(0,\cdots,0)}\\
        =&\frac{(-it)^m}{m!}\frac{\partial^m}{\partial z_{X_1}\cdots\partial z_{X_m}}[\underbrace{H,[H,\cdots [H}_{m}, O_i]\cdots]]\big|_{z=(0,\cdots,0)}\\
        =&\frac{(-it)^m}{m!}\sum_{\sigma\in \mathcal{P}_m}[\partial_{z_{X_{\sigma(1)}}}H,\cdots[\partial_{z_{X_{\sigma(m)}}}H,O_i]\cdots]\big|_{z=(0,\cdots,0)}\\
        =&\frac{1}{m!}\sum_{\sigma\in \mathcal{P}_m}[h_{X_{\sigma(1)}},\cdots[h_{X_{\sigma(m)}},O_i]\cdots].
    \end{split}
\end{eqnarray}
As a result, the cluster expansion of the single-step Hamiltonian dynamics can be written by
\begin{align}
    e^{iHt}O_ie^{-iHt}=\sum_{m\geq 0}^{+\infty}\sum_{\bm{W}\in\mathcal{C}_m}\frac{\bm\lambda^{\bm{W}}}{\bm{W}!}\frac{(-it)^m}{m!}\sum\limits_{\sigma\in \mathcal{P}_m}\left[h_{W_{\sigma(1)}},\cdots [h_{W_{\sigma(m)}},O_i]\right].
\end{align}
Here $\mathcal{P}_m$ represents the permutation group on the set $\{1,\cdots,m\}$. We denote the cluster derivative
\begin{align}
    D_{\bm{W}}[e^{iHt}O_ie^{-iHt}]=\frac{(-it)^m}{m!}\sum\limits_{\sigma\in \mathcal{P}_m}\left[h_{W_{\sigma(1)}},\cdots [h_{W_{\sigma(m)}},O_i]\right].
    \label{Eq:gradient}
\end{align}
From Eq.~\eqref{Eq:gradient}, we know that $W_{\sigma(1)}\cap W_{\sigma(2)}\cap\cdots\cap W_{\sigma(m)}\cap {\rm supp}(O_i)=\emptyset$ may result in $D_{\bm{W}}\left[e^{iHt}O_ie^{-iHt}\right]=0$~\cite{wild2023classical}. This property dramatically reduces the computational complexity in approximating $e^{iHt}O_ie^{-iHt}$, which only needs to consider connected clusters $\bm{W}$ with bounded size. 

In this article, we consider the $K$-step scenario driven by $\{H^{(1)},\cdots,H^{(K)}\}$ and corresponding time parameters $\{t_1,\cdots,t_K\}$. According to the linear property of the commute net, for any Hermitian operator $A$, we have 
\begin{align}
    \left[A,\sum\limits_{\sigma\in \mathcal{P}_m}[h_{W_{\sigma(1)}},\cdots[h_{W_{\sigma(m)}},O_i]]\right]= \sum\limits_{\sigma\in \mathcal{P}_m}\left[A,[h_{W_{\sigma(1)}},\cdots[h_{W_{\sigma(m)}},O_i]]\right].
    \label{Eq:linearproperty}
\end{align}

We first consider the cluster expansion of $2$-step Hamiltonian dynamics
\begin{eqnarray}
    \begin{split}
        &e^{iH^{(2)}t_2}e^{iH^{(1)}t_1}O_ie^{-iH^{(1)}t_1}e^{-iH^{(2)}t_2}\\
        =&\sum\limits_{m_2\geq 0}\sum\limits_{\bm{W}_2\in\mathcal{C}_{m_2}}\frac{\bm\lambda^{\bm{W}_2}}{\bm{W}_2!}D_{\bm{W}_2}\left[e^{iH^{(2)}t_2}\sum_{m_1\geq 0}^{+\infty}\sum_{\bm{W}_1\in\mathcal{C}_{m_1}}\frac{\bm\lambda_1^{\bm{W}_1}}{\bm{W}_1!}\frac{(-it)^{m_1}}{m_1!}\sum\limits_{\sigma\in \mathcal{P}_{m_1}}\left[h_{W_{\sigma(1)}},\cdots [h_{W_{\sigma(m_1)}},O_i]\right]e^{-iH^{(2)}t_2}\right]\\
        =&\sum\limits_{m_2\geq 0}\sum\limits_{\bm{W}_2\in\mathcal{C}_{m_2}}\frac{\bm\lambda^{\bm{W}_2}}{\bm{W}_2!}\frac{(-it_2)^{m_2}}{m_2!}\sum\limits_{\sigma_2\in \mathcal{P}_{m_2}}\left[h_{W_{\sigma_2(1)}}\cdots\left[h_{W_{\sigma_2(m_2)}},\sum_{m_1\geq 0}^{+\infty}\sum_{\bm{W}_1\in\mathcal{C}_{m_1}}\frac{\bm\lambda_1^{\bm{W}_1}}{\bm{W}_1!}\frac{(-it)^{m_1}}{m_1!}\sum\limits_{\sigma\in \mathcal{P}_{m_1}}\left[h_{W_{\sigma(1)}},\cdots [h_{W_{\sigma(m_1)}},O_i]\right]\right]\right]\\
        =&\sum\limits_{m_1,m_2\geq 0}\sum\limits_{(\bm{W}_1,\bm{W}_2)}\frac{\bm\lambda^{\bm{W}_1}\bm\lambda^{\bm{W}_2}}{\bm{W}_1!\bm{W}_2!}\frac{(-it_1)^{m_1}(-it_2)^{m_2}}{m_1!m_2!}\sum\limits_{\substack{\sigma_1\in \mathcal{P}_{m_1}\\\sigma_2\in \mathcal{P}_{m_2}}}\left[h_{W_{\sigma_1(1)}},\cdots \left[h_{W_{\sigma_1(m_1)}}\cdots\left[h_{W_{\sigma_2(m_2)}},O_i\right]\right]\right],
    \end{split}
\end{eqnarray}
where the second equality comes from the relationship given by Eq.~\eqref{Eq:linearproperty}. Repeat above process for $K$ times, we have the cluster expansion of $K$-step Hamiltonian dynamics, that is
\begin{align}
    U_{O_i}(\vec{t})=\sum\limits_{\substack{m_1\geq0\\\cdots\\m_K\geq 0}}\sum\limits_{\substack{\bm{W}_1\in\mathcal{C}_{m_1}\\\cdots\\\bm{W}_K\in\mathcal{C}_{m_K}}}\frac{\prod_{k=1}^K(\bm\lambda^{\bm{W}_k}(-it_k)^{m_k})}{\prod_{k=1}^K\bm{W}_k!m_k!}\sum\limits_{\substack{\sigma_1\in \mathcal{P}_{m_1}\\\cdots\\\sigma_K\in \mathcal{P}_{m_K}}}\left[h_{W_{\sigma_1(1)}},\cdots \left[h_{W_{\sigma_1(m_1)}},\cdots\left[h_{W_{\sigma_K(m_K)}},O_i\right]\right]\right].
    \label{Eq:Ut}
\end{align}
Here, notations $\sigma_1,\cdots,\sigma_K$ represent $K$ permutations, and $\mathcal{P}_{m_1}\cdots,\mathcal{P}_{m_K}$ represents corresponding permutation groups.

Similar to the single-step Hamiltonian dynamics, we know that if clusters $\bm{W}_1,\bm{W}_2,\cdots, \bm{W}_K$ and $O_i$ are disconnected, then the commute net $\left[h_{W_{\sigma_1(1)}},\cdots [h_{W_{\sigma_K(m_K)}},O_i]\right]=0$, which can be summarized as the following lemma.
\begin{lemma}
    Given clusters $\bm{W}_1,\bm{W}_2,\cdots,\bm{W}_K$ and an observable $O_i$,
    if the supper-interaction graph induced by $\bm{W}=(\bm{W}_1,\bm{W}_2,\cdots,\bm{W}_K,O_i)$ is disconnected, then the commute net $$\left[h_{W_{\sigma_1(1)}},\cdots[h_{W_{\sigma_1(m_1)}}\cdots[h_{W_{\sigma_K(1)}}\cdots[h_{W_{\sigma_K(m_K)}},O_i]]]\right]=0,$$
    where $\abs{\bm{W}_k}=m_k$ and $\sigma_k(1),\sigma_k(2),\cdots,\sigma_k(m_k)$ represents an entry of the permutation group $\mathcal{P}_{m_k}$.
\end{lemma}
\begin{proof}
    Denote all connected super-interaction graph as  $\mathcal{G}_{\bm{W}_1,\bm{W}_2,\cdots,\bm{W}_K,O_i}^K$. Consider a cluster $\bm W\notin \mathcal{G}_{\bm{W}_1,\bm{W}_2,\cdots,\bm{W}_K,O_i}^K$. For every permutation series $(\sigma_1(1),\cdots,\sigma_1(m_1),\cdots,\sigma_K(m_K))$, there exists an index $\sigma_k(s)$ such that $\bm W_{\sigma_k(s)}$ and $\bm W_{\sigma_k(s+1)}\cup\cdots\cup\bm W_{\sigma_K(m_K)}\cup {\rm supp}(O_i)$ does not have an overlap. This directly results in 
    \begin{align}
        \left[h_{W_{\sigma_{k(s)}}},\cdots[h_{W_{\sigma_{k(m_k)}}}\cdots[h_{W_{\sigma_K(1)}}\cdots[h_{V_{\sigma_K(m_K)}},O_i]]]\right]=0,
    \end{align}
    and the concerned commutator vanishes.
\end{proof}

Using this property, we may rewrite the above expression by introducing the connected cluster set $\mathcal{G}_m^{K,O_i}$ composed by all connected super-interaction graphs $G^K_{\bm{W}_1,\cdots, \bm{W}_K}$ (connected to $O_i$) with size $$m=\abs{\bm{W}_1}+\cdots+\abs{\bm{W}_K}.$$ Here, $O_i$ is a single-qubit operator non-trivially performs on qubit $i$, then $\{\bm{W}_1,\cdots,\bm{W}_K,O_i\}$ are connected implies ${\rm supp}(O_i)\in \bm{W}_k$ for some $k\in[K]$. Such observation enables us to only consider summation over $\mathcal{G}_m^{K,O_i}$, meanwhile truncate the cluster expansion up to $M$ order, that is
\begin{align}
    V_{O_i}(\vec{t})=\sum\limits_{\substack{m_1\geq0\\\cdots\\m_K\geq 0}}^M\sum\limits_{\bm{W}_1\cdots, \bm{W}_K\in\mathcal{G}_m^{K,O_i}}\frac{\prod_{k=1}^K(\bm\lambda^{\bm{W}_k}(-it_k)^{m_k})}{\prod_{k=1}^K\bm{W}_k!m_k!}\sum\limits_{\substack{\sigma_1\in \mathcal{P}_{m_1}\\\cdots\\\sigma_K\in \mathcal{P}_{m_K}}}\left[h_{W_{\sigma_1(1)}},\cdots [h_{W_{\sigma_K(m_K)}},O_i]\right].
    \label{Eq:clusterexpansion}
\end{align}

\subsection{Support size evaluation}
\label{App:support}

\begin{lemma}
Consider the Hamiltonian dynamics with evolution time $t=\max_k\{\abs{t_k}\}$, the operator $U_{O_i}(\vec{t})$ can be approximated by $V_{O_i}(\vec{t})=\sum_{Q(O_i)}\alpha_{Q(O_i)}Q(O_i)$ such that $\|\mathcal{U}_{O_i}(\vec{t})-\mathcal{V}_{O_i}(\vec{t})\|_{\diamond}\leq\epsilon^{\prime}\|O_i\|_{\infty}$. Here, $\max_{Q(O_i)}\abs{{\rm supp}(Q(O_i))}\leq \mathcal{O}(M(t))$ and $V_{O_i}(\vec{t})$ contains $L\leq\mathcal{O}\left(\left(e\mathfrak{d}\right)^{M(t)}\right)$ Pauli operators $Q(O_i)$ with 
\begin{eqnarray}
M(t)=\left\{
    \begin{split}
        &\frac{\log(1/\epsilon^{\prime})-K\log(1-2teK\mathfrak{d})}{K\log(1/(2teK\mathfrak{d}))}, ~\text{when}~t<\frac{1}{2eK\mathfrak{d}}\\
        &e^{\pi teK\mathfrak{d}}\log\left[\frac{e^{\pi teK\mathfrak{d}}}{\epsilon^{\prime}}\right], ~\text{when}~t=\mathcal{O}(1)
    \end{split}
\right.
\end{eqnarray}
and $\mathcal{U}_{O_i}(\vec{t})$, $\mathcal{V}_{O_i}(\vec{t})$ are channel representations of $U_{O_i}(\vec{t})$ and $V_{O_i}(\vec{t})$, respectively.
\end{lemma}

\begin{theorem}[Case~1: $\abs{t}<1/2eK\mathfrak{d}$]
\label{them:localapproxcase1}
Suppose a single qubit observable $O_i$, $K$-step quantum dynamics driven by $\Lambda$-local Hamiltonians $\{H^{(1)},\cdots,H^{(K)}\}$ and corresponding constant time parameters $\vec{t}=\{t_1,\cdots,t_K\}$. If the evolution time $\max_k\{\abs{t_k}\}=t<1/(2eK\mathfrak{d})$, then the operator $U_{O_i}(\vec{t})=\prod_{k=1}^Ke^{iH^{(k)}t_k}O_i\prod_{k=1}^Ke^{-iH^{(k)}t_k}$ can be approximated by $V_{O_i}(\vec{t})$ (Eq.~\eqref{Eq:clusterexpansion}) such that 
\begin{align}
    \|U_{O_i}(\vec{t})-V_{O_i}(\vec{t})\|_{\infty}\leq\epsilon\|O_i\|_{\infty},
\end{align}
where the number of involved cluster terms 
   \begin{align}
       M(t)=\frac{\log(1/\epsilon)-K\log(1-2teK\mathfrak{d})}{K\log(1/(2teK\mathfrak{d}))}.
        \label{Eq:Msmall}
   \end{align}
\end{theorem}

\begin{proof}
     Let $t=\max_{k\in[K]}\{t_k\}$ such that $\abs{t}\leq 1/(2eK\mathfrak{d})$, where the constant $\mathfrak{d}$ represents the maximum degree of the Hamiltonian interaction graph. ($K\mathfrak{d}$ represents the maximum degree of the $m$-sized graph $G^K_{\bm{W}_1\cdots\bm{W}_K}$ which connects to $O_i$.) Now we study the convergence of the cluster expansion
  $$U_{O_i}(\vec{t})=\sum\limits_{\substack{m_1\geq0\\\cdots\\m_K\geq 0}}\sum\limits_{\bm{W}_1\cdots, \bm{W}_K\in\mathcal{G}_m^{K,O_i}}\frac{\prod_{k=1}^K(\bm\lambda^{\bm{W}_k}(-it_k)^{m_k})}{\prod_{k=1}^K\bm{W}_k!m_k!}\sum\limits_{\substack{\sigma_1\in \mathcal{P}_{m_1}\\\cdots\\\sigma_K\in \mathcal{P}_{m_K}}}\left[h_{W_{\sigma_1(1)}},\cdots \left[h_{W_{\sigma_1(m_1)}}\cdots\left[h_{W_{\sigma_K(m_K)}},O_i\right]\right]\right]$$
up to index $m_1,m_2,\cdots m_K\leq M$. Without loss of generality, we assume Hamiltonians $H^{(k)}$ are $\mathcal{O}(1)$-local, that is $\max\abs{{\rm supp}(h_X)}=\Lambda\leq\mathcal{O}(1)$ and $\|h_X\|_{\infty}\leq 1$. Let $m=m_1+\cdots+m_K$, we have
 \begin{eqnarray}
        \begin{split}
       \epsilon_M(\vec{t})=&\|U_{O_i}(\vec{t})-V_{O_i}(\vec{t})\|_{\infty}\\
       =&\left\|\sum\limits_{\substack{m_1\geq M+1\\\cdots\\m_K\geq M+1}}\sum\limits_{\bm{W}_1\cdots, \bm{W}_K\in\mathcal{G}_m^{K,O_i}}\frac{\prod_{k=1}^K(\bm\lambda^{\bm{W}_k}(-it_k)^{m_k})}{\prod_{k=1}^K\bm{W}_k!m_k!}\sum\limits_{\substack{\sigma_1\in \mathcal{P}_{m_1}\\\cdots\\\sigma_K\in \mathcal{P}_{m_K}}}\left[h_{W_{\sigma_1(1)}},\cdots [h_{W_{\sigma_K(m_K)}},O_i]\right]\right\|_{\infty}\\
        \leq &\sum\limits_{m_1,\cdots,m_K\geq M+1}\sum\limits_{\bm{W}_1\cdots, \bm{W}_K\in\mathcal{G}_m^{K,O_i}}\frac{{\bm\lambda}^{\bm{W}_1}\cdots {\bm\lambda}^{\bm{W}_K}(2t_1)^{m_1}\cdots (2t_K)^{m_K}}{(\bm{W}_1!\cdots\bm{W}_K!)}\left\|O_i\right\|_{\infty}\\
        \leq&\sum\limits_{m_1,\cdots,m_K\geq M+1}(2t_1)^{m_1}\cdots(2t_K)^{m_K}\abs{\mathcal{G}_{m}^{K,O_i}}\|O_i\|_{\infty}\\
        \leq&\|O_i\|_{\infty}\sum\limits_{m_1,\cdots,m_K\geq M+1}(2t_1)^{m_1}\cdots(2t_K)^{m_K}\abs{eK\mathfrak{d}}^{m_1+\cdots m_K}\\
        \leq &\|O_i\|_{\infty}\left[\sum\limits_{l\geq M+1}(2teK\mathfrak{d})^{l}\right]^K.
        \end{split}
        \label{Eq:app}
    \end{eqnarray}
The second line is valid since $\left\|\left[h_{W_{\sigma_1(1)}},\cdots [h_{W_{\sigma_K(m_K)}},O_i\right]\right\|_{\infty}\leq 2^{m_1+\cdots +m_K}\left(\max\|h_i\|_{\infty}\right)^m\|O_i\|_{\infty}\leq 2^{m}$, and the fifth line comes from $\abs{\mathcal{G}_m^{K,O_i}}\leq (eK\mathfrak{d})^m$.

As a result, when $t<1/(2eK\mathfrak{d})$, we have 
\begin{align}
    \epsilon_M(\vec{t})\leq \|O_i\|_{\infty}\frac{(2teK\mathfrak{d})^{K(M+1)}}{(1-2teK\mathfrak{d})^K}.
\end{align}
Let $\epsilon=\frac{(2teK\mathfrak{d})^{K(M+1)}}{(1-2teK\mathfrak{d})^K}$ and consider $\|O_i\|_{\infty}=1$ for $O_i\in\{I,X,Y,Z\}$, these result in 
\begin{align}
    M(t)=\frac{\log(1/\epsilon)-K\log(1-2teK\mathfrak{d})}{K\log(1/(2teK\mathfrak{d}))}.
\end{align}

Finally, $\epsilon_M(\vec{t})\leq\epsilon\|O_i\|_{\infty}$ implies 
\begin{align}
    \|\mathcal{U}_{O_i}(\vec{t})-\mathcal{V}_{O_i}(\vec{t})\|_{\diamond}\leq\epsilon.
\end{align}
\end{proof}

Above result demonstrate that when the evolution time $\abs{\max_k\{t_k\}}$ is less than a constant threshold $t^*=1/(2eK\mathfrak{d})$, the approximation $V_{O_i}(\vec{t})$ (Eq.~\eqref{Eq:clusterexpansion}) may provide an estimation to $U_i(\vec{t})$ in the context of the $\|\cdot\|_{\infty}$ norm, where $V_i(\vec{t})$ can be decomposed by a series of operators that acts on $\mathcal{O}(KM(t))$ qubits. In the following, we demonstrate that for more general evolution time $\abs{t}=\mathcal{O}(1)$, Eq.~\eqref{Eq:clusterexpansion} can also provide an estimation to $U_i(\vec{t})$ in the context of the $\|\cdot\|_{\infty}$ norm.

\begin{theorem}[Case~2: $\abs{t}=\mathcal{O}(1)$]\label{them:localapproxcase2}
Suppose we are given a single qubit observable $O_i$, $K$-step quantum dynamics driven by $\Lambda$-local Hamiltonians $\{H^{(1)},\cdots,H^{(K)}\}$ and corresponding constant time parameters $\vec{t}=\{t_1,\cdots,t_K\}$. If the evolution time $t=\max_k\{t_k\}=\mathcal{O}(1)$, then the operator $U_{O_i}(\vec{t})=\prod_{k=1}^Ke^{iH^{(k)}t_k}O_i\prod_{k=1}^Ke^{-iH^{(k)}t_k}$ can be approximated by $V_{O_i}(\vec{t})$ (Eq.~\eqref{Eq:clusterexpansion}) such that 
\begin{align}
    \|U_{O_i}(\vec{t})-V_{O_i}(\vec{t})\|_{\infty}\leq\epsilon\|O_i\|_{\infty},
\end{align}
where the number of involved cluster terms 
   \begin{align}
   M(t)=e^{\pi teK\mathfrak{d}/\kappa}\log\left[\frac{1}{\epsilon}\frac{e^{\pi teK\mathfrak{d}/\kappa}-1}{(1-\kappa)^K}\right],
   \label{Eq:Mlarge}
   \end{align}
with the parameter $\kappa\in\mathcal{O}(1)$.
\end{theorem}

\noindent\textbf{Proof of Theorem~\ref{them:localapproxcase2}:}
Noting that above process can be further generalized to an arbitrary constant time $t$ by means of analytic continuation. Consider the radius of a disk $R>1$, the analytic continuation can be achieved by using the map $t\mapsto t\phi(z)$, where the complex function $$\phi(z)=\frac{\log(1-z/R^{\prime})}{\log(1-1/R^{\prime})}$$ maps a disk onto an elongated region along the real axis~\cite{wild2023classical}. Here, the parameter $R^{\prime}>R$, and $\phi(z)$ is analytic on the closed desk $D_R=\{z\in\mathbb{C}:\abs{z}\leq R\}$. Meanwhile, $\phi(z)$ satisfies $\phi(0)=0$, $\phi(1)=1$ and we select the branch ${\rm Im}(\phi(z))\leq -\pi/(2\log(1-1/R^{\prime}))$.
    
    We consider the function 
     \begin{align}
         f(z)=\prod_{k=1}^Ke^{iH^{(k)}t_k\phi(z)}O_i\prod_{k=1}^Ke^{-iH^{(k)}t_k\phi(z)}
     \end{align}
     on the region $\abs{z}\leq sR$ where $s\in(0,1)$. Consider a curve ${\mathcal{C}}^{\prime}=\{\abs{w}=R\}$, according to the Cauchy integral method, we have
    \begin{eqnarray}
    \begin{split}
         f(z)=&\frac{1}{2\pi i}\oint_{{\mathcal{C}}^{\prime}}\frac{f(w)}{w-z}dw\\
         =&\frac{1}{2\pi i}\oint_{{\mathcal{C}}^{\prime}}\frac{f(w)}{w}\left(1-\frac{z}{w}\right)^{-1}dw\\
         =&\frac{1}{2\pi i}\oint_{{\mathcal{C}}^{\prime}}\frac{f(w)}{w}\left(\sum\limits_{k=0}^M\left(\frac{z}{w}\right)^k+\left(\frac{z}{w}\right)^M\left(1-\frac{z}{w}\right)^{-1}\right)dw\\
         =&\sum\limits_{k=0}^M\frac{1}{2\pi i}\oint_{{\mathcal{C}}^{\prime}}\frac{f(w)}{w^k}z^k+\frac{1}{2\pi i}\oint_{{\mathcal{C}}^{\prime}}\frac{f(w)}{w-z}\left(\frac{z}{w}\right)^{M+1}dw\\
         =&\sum\limits_{k=0}^M\frac{f^{(k)}(0)}{k!}z^k+\frac{1}{2\pi i}\oint_{{\mathcal{C}}^{\prime}}\frac{f(w)}{w-z}\left(\frac{z}{w}\right)^{M+1}dw.
    \end{split}
    \end{eqnarray}
As a result, the truncated error can be upper bounded by
 \begin{eqnarray}
    \begin{split}
    \left\|f(z)-\sum\limits_{k=0}^M\frac{f^{(k)}(0)}{k!}z^k\right\|_{\infty}=& \left\|\frac{1}{2\pi i}\oint_{{\mathcal{C}}^{\prime}}\frac{f(w)}{w-z}\left(\frac{z}{w}\right)^{M+1}dw\right\|_{\infty}\\
    \leq&\frac{1}{2\pi}\oint_{{\mathcal{C}}^{\prime}}\frac{\|f(w)\|_{\infty}}{\|w-z\|}\left\|\frac{z}{w}\right\|^{M+1}dw.
\end{split}
\end{eqnarray}
We require the following result to evaluate the upper bound of $\|f(w)\|_{\infty}$.

\begin{definition}[Multi-variable complex analytic function]
    Suppose $g:D\mapsto\mathbb{C}$ be a function on the domain $D\subset\mathbb{C}^K$, if for any vector $\beta\in D$, there exists a $r$-radius cylinder $P_K(\beta,r)$ centered on $\beta$, such that 
    \begin{align}
        g(z)=\sum\limits_{\alpha_1,\cdots,\alpha_K\geq 0}c_{\vec{\alpha}}(z_1-\beta_1)^{\alpha_1}\cdots(z_K-\beta_K)^{\alpha_K},
    \end{align}
    then $g$ is analytic on the point $\beta=(\beta_1,\cdots,\beta_K)$.
\end{definition}

\begin{lemma}
   Given complex values $\vec{w}=(w_1,\cdots,w_K) \in\mathbb{C}^K$, if ${\rm Im}(w_k)\leq 1/(2eK\mathfrak{d})$ for all $k\in[K]$, we have
    \begin{align}
        \|U_i(\vec{w})\|\leq \frac{\|O_i\|}{(1-2\abs{\max_k{\rm Im}(w_k)}eK\mathfrak{d})^K},
    \end{align}
    where $\mathfrak{d}$ represents the maximum degree of the interaction graph induced by Hamiltonian $H$.
    \label{lemma:upperboundnorm}
\end{lemma}
\begin{proof}
     Eq.~\eqref{Eq:app} provides an approximation to $U_i(\vec{t})$ when $\max_k\abs{t_k}\leq1/(2eK\mathfrak{d})$, in other word, $U_i(\vec{t})$ remains analytic for all complex values $t_k\in\mathbb{C}$ in the range $\abs{t_k}<1/(2eK\mathfrak{d})$. Specifically, given any $\beta_1,\beta_2,\cdots,\beta_K\in\mathbb{R}$, we may write $U_i(\vec{t})=\prod_{k=1}^Ke^{iH^{(k)}(t_k-\beta_k)}e^{iH^{(k)}\beta_k}O_i\prod_{k=1}^Ke^{-iH^{(k)}(t_k-\beta_k)}e^{-iH^{(k)}\beta_k}$. Equivalently, we have
     \begin{align}
         U_i(\vec{t})=\sum\limits_{l_1\cdots,l_K\geq 0}u_{l_1,\cdots,l_K}(t_1-\beta_1)^{l_1}\cdots(t_K-\beta_K)^{l_K}
     \end{align}
     for some operators $u_{l_1,\cdots,l_K}$, which naturally implies $U_i(\vec{t})$ is analytic for all complex values of $\vec{t}$ on a disk in the complex plane of radius $1/(2eK\mathfrak{d})$ around any point on the real axis. 

For $\vec{w}=(w_1,\cdots,w_K)\in\mathbb{C}^K$, noting that $e^{-i(w_k-{\rm Re}(w_k))H^{(k)}}e^{i(w_k-{\rm Re}(w_k))H^{(k)}}=I$, then for any matrix $A$, the matrix $$e^{-i(w_k-{\rm Re}(w_k))H^{(k)}}Ae^{i(w_k-{\rm Re}(w_k))H^{(k)}}$$ is similar to $A$, and they thus share the same spectrum information. Although this property may not be directly applied to the current case, we note that when $\abs{{\rm Im}(w_k)}<\beta^*\approx\ln 4/\mathfrak{d}$ and $\|H^{(k)}\|=\mathcal{O}(\mathfrak{d}n)$,
\begin{align}
    \|e^{-i(w_k-{\rm Re}(w_k))H^{(k)}}UAU^{\dagger}e^{i(w_k-{\rm Re}(w_k))H^{(k)}}\|\leq\|e^{-i(w_k-{\rm Re}(w_k))H^{(k)}}Ae^{i(w_k-{\rm Re}(w_k))H^{(k)}}\|
    \label{Eq:unitary}
\end{align}
for random unitary matrix $U$ with large probability. From a high-level perspective, this relationship is valid since the random unitary vanishes large-weight operators. Specifically, we choose an \emph{arbitrary} quantum state $|\psi\rangle$ and consider an approximately unitary $2$-design ensemble $U\sim \mathcal{U}_2$, and we have
\begin{eqnarray}
\begin{split}
     &\mathbb{E}_{U\sim\mathcal{U}_2}\abs{\langle\psi|e^{-i{\rm Im}(w_k)H^{(k)}}UAU^{\dagger}e^{i{\rm Im}(w_k)H^{(k)}}|\psi\rangle}^2\\
     =&\mathbb{E}_{U\sim\mathcal{U}_2}{\rm Tr}\left[e^{i{\rm Im}(w_k)H^{(k)}}|\psi\rangle\langle\psi|e^{-i{\rm Im}(w_k)H^{(k)}}UAU^{\dagger}\right]{\rm Tr}\left[e^{i{\rm Im}(w_k)H^{(k)}}|\psi\rangle\langle\psi|e^{-i{\rm Im}(w_k)H^{(k)}}UAU^{\dagger}\right]\\
     \leq&\frac{{\rm Tr}(A^2)}{2^n(2^n+1)}\left(\langle\psi|e^{-i{\rm Im}(w_k)H^{(k)}}|\psi\rangle\langle\psi|e^{i{\rm Im}(w_k)H^{(k)}}|\psi\rangle\right)\\
     \leq & {\rm Tr}(A^2)\left(\frac{e^{\abs{{\rm Im}(w_k)\mathfrak{d}}}}{4}\right)^n.
\end{split}
\end{eqnarray}
where the third line comes from Lemma~3 in Supp material of Ref.~\cite{cerezo2021cost} and the fourth line comes from the assumption $\|H^{(k)}\|\leq\mathcal{O}(\mathfrak{d}n)$. As a result, for any quantum state $|\psi\rangle$ and $\beta^*=\ln 4/\mathfrak{d}$, the $\abs{\langle\psi|e^{-i{\rm Im}(w_k)H^{(k)}}UAU^{\dagger}e^{i{\rm Im}(w_k)H^{(k)}}|\psi\rangle}$ is upper bounded by a constant value with nearly unit probability (promised by Markov inequality). Noting that this property holds for \emph{any quantum state} $|\psi\rangle$, as a result, $$\|e^{-i(w_k-{\rm Re}(w_k))H^{(k)}}UAU^{\dagger}e^{i(w_k-{\rm Re}(w_k))H^{(k)}}\|^2={\rm sup}_{|\psi\rangle}\abs{\langle\psi|e^{-i{\rm Im}(w_k)H^{(k)}}UAU^{\dagger}e^{i{\rm Im}(w_k)H^{(k)}}|\psi\rangle}$$ should also be upper bounded by a constant value with large probability. On other hand, it is well known that $e^{\beta H}$ may dramatically increase $\|e^{\beta H}Ae^{-\beta H}\|$ even for constant $\beta$. Then it is reasonable to assume $\|e^{-i(w_k-{\rm Re}(w_k))H^{(k)}}Ae^{i(w_k-{\rm Re}(w_k))H^{(k)}}\|>\bm w(1)$. These two results finally give rise to inequality~\ref{Eq:unitary} which completes the reduction from $K$ Hamiltonians dynamics to single Hamiltonian dynamics studied in Ref.~\cite{wild2023classical}.

We note that Ref.~\cite{schuster2024random} indicated that ${\rm poly}\log(n)$-depth quantum circuit suffices to approximate unitary $t$-design ensemble. This provides theoretical foundations in applying inequality~\ref{Eq:unitary} to constant time Hamiltonian dynamics. For any $w\in\mathbb{C}^K$, we have
\begin{eqnarray}
    \begin{split}
       \|U_i(\vec{w})\|=&\Big\|e^{-i(w_K-{\rm Re}(w_K))H^{(K)}}e^{-i{\rm Re}(w_K)H^{(K)}}\cdots e^{-i(w_1-{\rm Re}(w_1))H^{(1)}}e^{-i{\rm Re}(w_1)H^{(1)}}O_i\\
        &e^{i{\rm Re}(w_1)H^{(1)}}e^{i(w_1-{\rm Re}(w_1))H^{(1)}}\cdots e^{i{\rm Re}(w_K)H^{(K)}}e^{i(w_K-{\rm Re}(w_K))H^{(K)}}\Big\|\\
        \leq&\Big\|e^{-i(w_K-{\rm Re}(w_K))H^{(K)}}\cdots e^{-i(w_1-{\rm Re}(w_1))H^{(1)}}O_ie^{i(w_1-{\rm Re}(w_1))H^{(1)}}\cdots e^{i(w_K-{\rm Re}(w_K))H^{(K)}}\Big\|.
    \end{split}
\end{eqnarray}
For square matrices $A$ and $B$, the BCH expansion enables us to write the cluster expansion to $e^{tA}Be^{-tA}$~\cite{haah2024learning} for $t\in\mathbb{R}$. As a result, we have
  \begin{eqnarray}
      \begin{split}
           \|U_i(\vec{w})\|\leq&\left\|\sum\limits_{\substack{m_1\geq0\\\cdots\\m_K\geq 0}}\sum\limits_{\bm{W}_1\cdots, \bm{W}_K\in\mathcal{G}_m^{K,O_i}}\frac{\prod_{k=1}^K(\bm\lambda^{\bm{W}_k}(-i(w_k-{\rm Re}(w_k)))^{m_k})}{\prod_{k=1}^K\bm{W}_k!m_k!}\sum\limits_{\substack{\sigma_1\in \mathcal{P}_{m_1}\\\cdots\\\sigma_K\in \mathcal{P}_{m_K}}}\left[h_{W_{\sigma_1(1)}},\cdots [h_{W_{\sigma_K(m_K)}},O_i]\right]\right\|\\
            \leq&\|O_i\|\sum\limits_{m_1,\cdots,m_K\geq 0}\abs{(2(w_1-{\rm Re}(w_1)))^{m_1}\cdots(2(w_K-{\rm Re}(w_K)))^{m_K}}\abs{eK\mathfrak{d}}^{m_1+\cdots m_K}\\
           =&\|O_i\|\sum\limits_{m_1,\cdots,m_K\geq 0}\abs{(2({\rm Im}(w_1)))^{m_1}\cdots(2({\rm Im}(w_K)))^{m_K}}\abs{eK\mathfrak{d}}^{m_1+\cdots m_K}\\
           =&\frac{\|O_i\|}{(1-2\abs{\max_k{\rm Im}(w_k)}eK\mathfrak{d})^K}.
      \end{split}
  \end{eqnarray}
\end{proof}

Recall that 
$$f(w)=\prod_{k=1}^Ke^{iH^{(k)}t_k\phi(w)}O_i\prod_{k=1}^Ke^{-iH^{(k)}t_k\phi(w)}$$
where $\vec{t}\in\mathbb{R}^K$ and ${\rm Im}(\phi(w))\leq -\pi/(2\log(1-1/R^{\prime}))$. Assign $\vec{t}\phi(w)$ to $\vec{w}$ given in Lemma~\ref{lemma:upperboundnorm}, then Lemma~\ref{lemma:upperboundnorm} implies 
\begin{eqnarray}
\begin{split}
     \|f(w)\|=\|U_i(\phi(w)\vec{t})\|&\leq \frac{\|O_i\|}{(1-2\abs{\max_k{\rm Im}(t_k\phi(w))}eK\mathfrak{d})^K}\\
     &\leq \frac{\|O_i\|}{(1+\pi teK\mathfrak{d}/(\log(1-1/R^{\prime})))^K}
     \label{Eq:fupperbound}
\end{split}
\end{eqnarray}
for all $w\in C^{\prime}=\{\abs{w}=R\}$. This further results in 
\begin{eqnarray}
    \begin{split}
    \left\|f(z)-\sum\limits_{k=0}^M\frac{f^{(k)}(0)}{k!}z^k\right\|_{\infty}=& \left\|\frac{1}{2\pi i}\oint_{{\mathcal{C}}^{\prime}}\frac{f(w)}{w-z}\left(\frac{z}{w}\right)^{M+1}dw\right\|_{\infty}\\
    \leq&\frac{1}{2\pi}\oint_{{\mathcal{C}}^{\prime}}\frac{\|f(w)\|_{\infty}}{\|w-z\|}\left\|\frac{z}{w}\right\|^{M+1}dw\\
    \leq&\max\{\|f(w)\|\}\frac{s^{M+1}}{(1-s)}
    \label{Eq:analyticerror}
\end{split}
\end{eqnarray}
where the last line follow from the fact that $\abs{w-z}\geq R(1-s)$, $\abs{z}\leq sR$ and $\|w\|=R$. Combine inequalities~\ref{Eq:fupperbound} and \ref{Eq:analyticerror}, we have
\begin{align}
    \left\|f(z)-\sum\limits_{k=0}^M\frac{f^{(k)}(0)}{k!}z^k\right\|\leq \frac{\|O_i\|s^{M+1}}{(1+\pi teK\mathfrak{d}/(\log(1-1/R^{\prime})))^K(1-s)}.
    \label{Eq:error1}
\end{align}
Let $\kappa=\frac{-\pi teK\mathfrak{d}}{\log(1-1/R^{\prime})}$, $R^{\prime}$ can be further expressed by
\begin{align}
    \frac{1}{R^{\prime}}=1-e^{-\pi teK\mathfrak{d}/\kappa}.
\end{align}
Since the parameter $R^{\prime}>R$, we can always select $R$ such that $(R^{\prime})^M(R^{\prime}-1)=2R^M(R-1)$ holds. Substitute this relationship into the approximation upper bound given by~\ref{Eq:error1} and assign $s=1/R$, we finally obtain
\begin{align}
    \epsilon= \frac{s^{M+1}}{\left(1+\frac{\pi teK\mathfrak{d}}{\log(1-1/R^{\prime})}\right)^K(1-s)}=\frac{1}{(1-\kappa)^K}\left(1-e^{-\pi teK\mathfrak{d}/\kappa}\right)^M\left(e^{\pi teK\mathfrak{d}/\kappa}-1\right).
\end{align}
This implies truncating at order
\begin{align}
    M(t)=\frac{\log\left[\frac{1}{\epsilon}\frac{e^{\pi teK\mathfrak{d}/\kappa}-1}{(1-\kappa)^K}\right]}{\log\left[e^{\pi teK\mathfrak{d}/\kappa}/\left(e^{\pi teK\mathfrak{d}/\kappa}-1\right)\right]}\approx e^{\pi teK\mathfrak{d}/\kappa}\log\left[\frac{1}{\epsilon}\frac{e^{\pi teK\mathfrak{d}/\kappa}-1}{(1-\kappa)^K}\right].
\end{align}

\section{Learning $U_i(\vec{t})$ via Randomized measurement dataset}
\label{App:E}
Above results imply that $V_{O_i}(\vec{t})$ can provide a $\epsilon$-close approximation to $U_{O_i}(\vec{t})=U^{\dagger}(\vec{t})O_iU(\vec{t})$ in the context of the operator norm. Specifically, we have
\begin{align}
    V_{O_i}(\vec{t})=\sum\limits_{l=1}^L\alpha_l Q_l(O_i),
    \label{Eq:observable}
\end{align}
where $\alpha_l$-s are real-valued coefficients, the operator $Q_l(O_i)$ represents a Pauli operator that connects to $O_i$, whose support size upper bound $M(t)$ is given by Eqs.~\ref{Eq:Msmall} and~\ref{Eq:Mlarge} ($\abs{{\rm supp}(Q_l(O_i))}\leq \mathcal{O}(M(t))$).

When the time evolution $\abs{t}=\mathcal{O}(1)$, the operator $V_{O_i}(\vec{t})$ may provide an approximation to $U_{O_i}(\vec{t})$ in the context of $\|\cdot\|_{\infty}$ distance, where the support size of $Q_l(O_i)$ logarithmically depends on $1/\epsilon$. When $\epsilon=\mathcal{O}(n^{-1})$ and parameters $K$, $\mathfrak{d}$ are constant values, $Q_l(O_i)$ only non-trivially performs on $\mathcal{O}(\log n)$ qubits, rather than ${\rm poly}\log n$ qubits, which is independent to the geometrical dimension of the Hamiltonian.

Furthermore, Eq.~\eqref{Eq:clusterexpansion} implies only connected path $(h_{W_1},h_{W_2},\cdots,h_{W_{KM}},O_i)$ resulting in non-zero commutator $[h_{W_{1}},\cdots [h_{W_{KM}},O_i]$. This property dramatically reduces the number of involved operators within $V_{O_i}(\vec{t})$. In fact, the number of connected paths in $V_{O_i}(\vec{t})$ (centered by $O_i$) is at most
\begin{align}
    \mathcal{O}\left(\left(e\mathfrak{d}\right)^{M(t)}\right).
\end{align}
This result can be verified in the following way. Starting from the qubit $O_i$, at which $\leq\mathfrak{d}$ operators act non-trivially, we generate all possible paths of size $\leq M(t)$ step by step, where each step adds a new connected operator $h_X$. Obviously, each step only has $\mathcal{O}(\mathfrak{d})$ choices, and each operator $h_X$ is constant local. This results in an upper bound on the number of $M(t)$-length connected paths, which further implies the number of Pauli operators within $V_{O_i}(\vec{t})$ may be upper bounded by
\begin{align}
    L\leq \mathcal{O}\left(4^{K\Lambda M(t)}\left(e\mathfrak{d}\right)^{M(t)}\right).\label{Eq:termupperbound}
\end{align}
This upper bound is valid since each connected path contains $KM(t)$ local terms, resulting in each path may cover $\mathcal{O}(K\Lambda M(t))$ qubits. Counting all the possible Pauli operators on such support may give rises to the factor $4^{K\Lambda M(t)}$. Here, $K$ represents the number of Hamiltonians in driving the quantum circuit $U(\vec{t})$, $\Lambda$ represents the locality of each Hermitian term $h_X$ and $\mathfrak{d}$ represents the maximum degree of the related interaction graph.

Given this observation, we can utilize the randomized dataset to learn coefficients $\alpha_l$.

\begin{definition}[Randomized measurement dataset for an unknown unitary] The learning algorithm accesses an unknown $n$-qubit unitary $U$ via a randomized measurement dataset of the following form,
\begin{align}
    \mathcal{T}_U(N)=\left\{|\psi_l\rangle=\bigotimes_{i=1}^n|\psi_{l,i}\rangle,|\phi_l\rangle=\bigotimes_{i=1}^n|\phi_{l,i}\rangle\right\}_{l=1}^N.
\end{align}
A randomized measurement dataset of size $N$ is constructed by obtaining $N$ samples from the unknown unitary $U$. One sample is obtained from one experiment given as follows.    
\end{definition}

\subsection{Short-Time Hamiltonian Dynamics}
\label{shorttime}
\begin{theorem}
\label{theorem6}
     Given an error $\epsilon$, failure probability $\delta$, an unknown $n$-qubit operator $U_{O_i}(\vec{t})$ with $\abs{t}<1/(2eK\mathfrak{d})$, which acts on a set of $M(t)$ qubits (given by Theorem~\ref{them:localapproxcase1}), and a dataset $\mathcal{T}_V=\{|\psi_l\rangle=\otimes_{i=1}^n|\psi_{l,i}\rangle,u_l\}_{l=1}^N$, where $|\psi_{l,i}\rangle$ is sampled uniformly from single-qubit stabilizer states, and $u_l$ is a random variable with $\mathbb{E}[u_l]=\langle\psi_l|U_{O_i}(\vec{t})|\psi_l\rangle$. Given a dataset size of 
    \begin{align}
        N=\frac{(4^{K\Lambda}3e\mathfrak{d})^{\mathcal{O}(M(t))}\log(1/\delta)}{\epsilon^2},
    \end{align}
    with probability $1-\delta$, we can learn an operator $V_{O_i}(\vec{t})$ such that $\|V_{O_i}(\vec{t})-U_{O_i}(\vec{t})\|_{\infty}\leq2\epsilon$.
\end{theorem}

\begin{proof}
The proof is fundamentally based on the following result.
\begin{fact}[Ref.~\cite{huang2024learning}]
    Let $U$ be a locally scrambling unitary. Then for all $P,Q\in\{I,X,Y,Z\}^{\otimes n}$, we have
    \begin{eqnarray}
    \mathbb{E}_U\left[U^{\dagger\otimes 2}(P\otimes Q)U^{\otimes 2}\right]=
    \left\{
        \begin{aligned}
        & 0, \ \text{if} \ P\neq Q,\\
        & \frac{1}{3^{\abs{P}}}\sum_{P\in\{I,X,Y,Z\}^{\otimes n}}\mathbb{E}_U\left[U^{\dagger\otimes 2}P^{\otimes 2}U^{\otimes 2}\right], \ \text{if} \ P=Q.\\
        \end{aligned}
    \right.
    \end{eqnarray}
\end{fact}
The orthogonality property immediately implies
\begin{eqnarray}
    \mathbb{E}_{C_i\sim {\rm Cl}(2)}\left[C_i^{\dagger\otimes 2}(P_i\otimes Q_i)C_i^{\otimes 2}\right]=
    \left\{
    \begin{aligned}
       & I^{\otimes 2}, \  \text{if} \ P_i=Q_i=I,\\
       & \frac{1}{3}\sum_{P_i\in\{X,Y,Z\}^{\otimes 2}}\left(P_i\otimes P_i\right), \ \text{if} \ P_i=Q_i\neq I.\\
       & 0, \ \text{if} \ P_i\neq Q_i.
        \end{aligned}
    \right.
    \label{Eq:cliffordprop}
\end{eqnarray}

Given the observable $V_{O_i}(\vec{t})$ (Eq.~\eqref{Eq:observable}), and let $Q_l(O_i)=\bigotimes_{j=1}^{M(t)}Q_l^{(j)}(O_i)$ with single-qubit Pauli operator $Q_l^{(j)}(O_i)\in\{I,X,Y,Z\}$, we can evaluate the mean value 
\begin{eqnarray}
\begin{split}
    &\mathbb{E}_{|\psi\rangle\sim {\rm Stab}_1^{\otimes n}}\langle\psi|V_{O_i}(\vec{t})|\psi\rangle\langle\psi|Q_l(O_i)|\psi\rangle\\
    =&\sum\limits_{k=1}^L\alpha_k \mathbb{E}_{|\psi\rangle\sim {\rm Stab}_1^{\otimes n}}\langle\psi|Q_k(O_i)|\psi\rangle\langle\psi|Q_{l}(O_i)|\psi\rangle\\
    =&\sum\limits_{k=1}^L\alpha_k\bigotimes_{j=1}^{M(t)}\mathbb{E}_{C_j\sim{\rm Cl}(2)}\langle0|C^{\dagger}_jQ_k^{(j)}(O_i)C_j|0\rangle\langle0|C^{\dagger}_jQ_l^{(j)}(O_i)C_j|0\rangle\\
    =&\frac{\alpha_l}{3^{\abs{Q_l(O_i)}}}\bigotimes_{j=1}^{M(t)}\sum_{P\in \{X,Y,Z\}}\langle0^2|P\otimes P|0^2\rangle\\
    =&\frac{\alpha_l}{3^{\abs{Q_l(O_i)}}}.
\end{split}
\end{eqnarray}
Equivalently, the coefficient 
\begin{align}
    \alpha_l=3^{\abs{Q_l(O_i)}}\mathbb{E}_{|\psi\rangle\sim {\rm Stab}_1^{\otimes n}}\langle\psi|V_{O_i}(\vec{t})|\psi\rangle\langle\psi|Q_l(O_i)|\psi\rangle
\end{align}
which can be learned by replacing the expectation with averaging over the randomized dataset $\mathcal{T}_V=\{|\psi_l\rangle,u_l\}$.

As a result, we can define the approximated observable $$\hat{V}_{O_i}(\vec{t})=\sum_{\abs{{\rm supp}(Q_l(O_i))}\leq M(t)}\hat{\alpha_l}Q_l(O_i),$$ where the learned coefficient
\begin{align}
    \hat{\alpha}_l=\frac{3^{\abs{Q_l(O_i)}}}{N}\sum\limits_{q=1}^Nu_q\langle\psi_q|Q_l(O_i)|\psi_q\rangle.
\end{align}
Here, one may doubt that why $u_q$ can be used to substitute $\langle\psi_q|V_{O_i}|\psi_q\rangle$. The reason stems from: we suppose $U_{O_i}(\vec{t})=\sum_Q\alpha_QQ$ with $\abs{{\rm supp}(Q)}\leq\mathcal{O}(\log^D(n))$, while we only concern parameters of $Q$ that appears in Eq.~\eqref{Eq:observable}. As a result, when the sample complexity
\begin{align}
    N=\frac{(4^{K\Lambda}3e\mathfrak{d})^{\mathcal{O}(M(t))}\log(1/\delta)}{\epsilon^2},
\end{align}
with probability $\geq 1-\delta$, the mean value $\alpha_l$ can be estimated by $\hat{\alpha_l}$ such that $\abs{\alpha_l-\hat{\alpha}_l}\leq\epsilon/(4^{K\Lambda}e\mathfrak{d})^{M(t)}$ promised by the Hoeffding's inequality and $\abs{Q_l(O_i)}\leq M(t)$. This further results in
\begin{eqnarray}
\begin{split}
    \left\|V_{O_i}(\vec{t})-\hat{V}_{O_i}(\vec{t})\right\|_{\infty}=&\left\|\sum\limits_{\abs{{\rm supp}(Q_l(O_i))}\leq M(t)}(\alpha_l-\hat{\alpha}_l)Q_l(O_i)\right\|_{\infty}\\
    \leq&\sum\limits_{\abs{{\rm supp}(Q_l(O_i))}\leq M(t)}\abs{\alpha_l-\hat{\alpha}_l}\|Q_l(O_i)\|_{\infty}\\
    \leq& (4^{K\Lambda}e\mathfrak{d})^{M(t)}\max\limits_{\abs{{\rm supp}(Q_l(O_i))}\leq M(t)}\abs{\alpha_l-\hat{\alpha}_l}\\
    \leq & \epsilon,
\end{split}
\end{eqnarray}
where the second line comes from the inequality~\ref{Eq:termupperbound} and the property of Pauli operator $\|Q_l(O_i)\|_{\infty}=1$.

Finally, using theorem~\ref{them:localapproxcase1}, we have the result
\begin{eqnarray}
    \begin{split}
        \left\|\hat{V}_{O_i}(\vec{t})-U_{O_i}(\vec{t})\right\|_{\infty}&\leq \left\|V_{O_i}(\vec{t})-U_{O_i}(\vec{t})\right\|_{\infty}+\left\|\hat{V}_{O_i}(\vec{t})-V_{O_i}(\vec{t})\right\|_{\infty}\\
        &\leq \epsilon\|O_i\|_{\infty}+\epsilon\\
        &=2\epsilon,
    \end{split}
\end{eqnarray}
where the second line comes from Theorem~\ref{them:localapproxcase1}.

\end{proof}

\subsection{Long-Time Hamiltonian Dynamics}
\label{longtime}

\begin{theorem}
     Given an error $\epsilon$, failure probability $\delta$, an unknown $n$-qubit observable $U_{O_i}(\vec{t})$ with $\abs{t}=\mathcal{O}(1)$, which acts on a set of $M(t)$ qubits (given by theorem~\ref{them:localapproxcase2}), and a dataset $\mathcal{T}_V=\{|\psi_l\rangle=\otimes_{i=1}^n|\psi_{l,i}\rangle,u_l\}_{l=1}^N$, where $|\psi_{l,i}\rangle$ is sampled uniformly from single-qubit stabilizer states, and $u_l$ is a random variable with $\mathbb{E}[u_l]=\langle\psi_l|U_{O_i}(\vec{t})|\psi_l\rangle$. Given a dataset size of 
    \begin{align}
        N=\frac{(4^{K\Lambda}3e\mathfrak{d})^{\mathcal{O}(M(t))}\log(1/\delta)}{\epsilon^2},
        \label{Eq:sample}
    \end{align}
    with probability $1-\delta$, we can learn an observable $\hat{V}_i(\vec{t})$ such that $\|\hat{V}_i(\vec{t})-U_i(\vec{t})\|_{\infty}\leq2\epsilon$.
\end{theorem}

\begin{proof}
    Define the approximated observable $\hat{U}_{O_i}(\vec{t})=\sum_{\abs{{\rm supp}(Q_l(O_i))}\leq M(t)}\hat{\alpha_l} Q_l(O_i)$, where the learned coefficient
\begin{align}
    \hat{\alpha}_l=\frac{3^{\abs{Q_l(O_i)}}}{N}\sum\limits_{q=1}^Nu_q\langle\psi_q|Q_l(O_i)|\psi_q\rangle.
\end{align}
When the sample complexity satisfies Eq.~\eqref{Eq:sample}, we still have $\abs{\alpha_l-\hat{\alpha}_l}\leq\epsilon/(4^{K\Lambda}e\mathfrak{d})^{M(t)}$ promised by the Hoeffding's inequality. Furthermore, we can upper bound
\begin{eqnarray}
    \begin{split}
        \left\|V_{O_i}(\vec{t})-\hat{V}_{O_i}(\vec{t})\right\|_{\infty}=&\left\|\sum\limits_{\abs{{\rm supp}(Q_l(O_i))}\leq M(t)}(\alpha_l-\hat{\alpha}_l)Q_l(O_i)\right\|_{\infty}\\
        \leq & \max_{l}\abs{\alpha_l-\hat{\alpha}_l}\sum\limits_{\abs{{\rm supp}(Q_l(O_i))}\leq M(t)}\|Q_l(O_i)\|_{\infty}\\
        \leq& \epsilon.
    \end{split}
\end{eqnarray}

Combine with Theorem~\ref{them:localapproxcase2}, we finally have
$\|\hat{V}_{O_i}(\vec{t})-U_{O_i}(\vec{t})\|_{\infty}\leq2\epsilon$. Here, we utilize the fact that $\|O\|_i=1$ for $O_i\in\{I,X,Y,Z\}$.
\end{proof}

\section{Sew all local evolution together}
\label{app:sew}
Recall that our learning algorithm is fundamentally based on the identity
$U(\vec{t})\otimes U^{\dagger}(\vec{t})=S\prod\limits_{i=1}^n\left[U^{\dagger}(\vec{t})S_iU(\vec{t})\right]$.
Substitute $S_i=\frac{1}{2}\sum_{O\in\{I,X,Y,Z\}}O_i\otimes O_{i+n}$ into above identity, the identity can be rewritten by
\begin{eqnarray}
    \begin{split}
         U(\vec{t})\otimes U^{\dagger}(\vec{t})&=S\prod\limits_{i=1}^n\left[\frac{1}{2}\sum\limits_{O\in\{I,X,Y,Z\}}U^{\dagger}(\vec{t})O_iU(\vec{t})\otimes O_{i+n}\right]\\
         &=S\prod\limits_{i=1}^n\left[\frac{1}{2}\sum\limits_{O\in\{I,X,Y,Z\}}U_{O_i}(\vec{t})\otimes O_{i+n}\right]\\
         &=S\prod\limits_{i=1}^nU_i(\vec{t}),
    \end{split}
\end{eqnarray}
with $U_i(\vec{t})=\frac{1}{2}\sum_{O\in\{I,X,Y,Z\}}U_{O_i}(\vec{t})\otimes O_{i+n}$.
In the following, we evaluate the distance between $U(\vec{t})\otimes U^{\dagger}(\vec{t})$ and 
\begin{align}
    \hat{V}^{\prime}=S\prod\limits_{i=1}^n\left[\frac{1}{2}\sum\limits_{O\in\{I,X,Y,Z\}}\hat{V}_{O_i}(\vec{t})\otimes O_{i+n}\right]=S\prod\limits_{i=1}^n\hat{V}_i(\vec{t})
\end{align}
with
\begin{align}
    \hat{V}_{O_i}(\vec{t})=\sum\limits_{\abs{{\rm supp}(Q_l(O_i))}\leq \mathcal{O}(M(t))}\hat{\alpha}_l Q_l(O_i),
\end{align}
and $\hat{V}_i(\vec{t})=\frac{1}{2}\sum_{O\in\{I,X,Y,Z\}}\hat{V}_{O_i}(\vec{t})\otimes O_{i+n}$. Let $\mathcal{U}(1)$, $\mathcal{U}(2)$, $\hat{\mathcal{V}}^{\prime}(2)$ be the channel representation of $U(\vec{t})$, $U(\vec{t})\otimes U^{\dagger}(\vec{t})$ and $\hat{V}^{\prime}$, meanwhile $\hat{\mathcal{V}}^{\prime}(1)={\rm Tr}_{>n}(\hat{\mathcal{V}}^{\prime}(2))$. Then we can upper bound the diamond distance by the similar approach given in Ref.~\cite{huang2024learning}:
\begin{eqnarray}
    \begin{split}
        &\left\|\hat{\mathcal{V}}^{\prime}(1)-\mathcal{U}(1)\right\|_{\diamond}
        =\left\|{\rm Tr}_{>n}\circ(\hat{\mathcal{V}}^{\prime}(2)-\mathcal{U}(2))\right\|_{\diamond}
        \leq\left\|{\rm Tr}_{>n}\right\|_{\diamond}\left\|\hat{\mathcal{V}}^{\prime}(2)-\mathcal{U}(2)\right\|_{\diamond}\\
        =&\left\|\hat{\mathcal{V}}_n(\vec{t})\cdots\hat{\mathcal{V}}_1(\vec{t})-\mathcal{U}_n(\vec{t})\cdots\mathcal{U}_1(\vec{t})\right\|_{\diamond}\\
        \leq&\sum\limits_{i=1}^n\left\|\hat{\mathcal{V}}_n\cdots\hat{\mathcal{V}}_{i+1}\mathcal{U}_i\cdots\mathcal{U}_1-\hat{\mathcal{V}}_n\cdots\hat{\mathcal{V}}_{i}\mathcal{U}_{i-1}\cdots\mathcal{U}_1\right\|_{\diamond}\\
        \leq&\sum\limits_{i=1}^n\|\hat{\mathcal{V}}_i-\mathcal{U}_i\|_{\diamond}\\
        \leq& 2\sum \limits_{i=1}^n\|\hat{V}_i(\vec{t})-U_i(\vec{t})\|_{\infty}\\
        \leq &2\sum\limits_{i=1}^n\sum\limits_{O\in\{X,Y,Z\}}\left\|\hat{V}_{O_i}(\vec{t})-U_{O_i}(\vec{t})\right\|_{\infty}\\
        \leq& 12n\epsilon.
    \end{split}
\end{eqnarray}
One may doubt that why the fifth line holds, given the fact that $\hat{\mathcal{V}}_i$ may not be a unitary channel. In the following section, we demonstrate how to approximate $\hat{V}_i(\vec{t})$ by using a unitary channel with $\mathcal{O}(\epsilon)$ additive error in terms of the diamond norm.

\section{Quantum Circuit Compilation}
\label{App:compile}
Given the obtained operator $V_{O_i}(\vec{t})=\sum_{\abs{{\rm supp}(P_l(O_i))}\leq\mathcal{O}(M(t))}\alpha_lP_l(O_i)$, it is required to be compiled into a corresponding quantum circuit. After the quantum learning phase, $V_i(\vec{t})=\sum_{O_i\in\{X,Y,Z\}}V_{O_i}(\vec{t})\otimes O_{i+n}$ is provided in the form of linear combinations of Pauli operators. The standard block-encoding method or the Linear Combination of Unitary~(LCU) generally require the classical post-selection. 

Observing the operator $V_i(\vec{t})$ essentially approximates the Hermitian operator $U_i(\vec{t})=U^{\dagger}(\vec{t})S_iU(\vec{t})$ whose eigenvalues only take values from $\{-1,+1\}$. As a result, $V_i(\vec{t})$ can be equivalently compiled by the Hamiltonian dynamics $e^{-\frac{i\pi}{2}\left(V_i(\vec{t})-I\right)}$. 

To verify this observation, we denote
\begin{align}
    J=\max\{\|V_i(\vec{t})-I\|_{\infty}, \|U_i(\vec{t})-I\|_{\infty}\},
\end{align}
and divide the evolution time $\pi/2$ into $J$ slices. Specifically, let the unitary channel $\mathcal{E}_{V_i}(\pi/2)=e^{-\frac{i\pi}{2}\left(V_i(\vec{t})-I\right)}(\cdot)e^{\frac{i\pi}{2}\left(V_i(\vec{t})-I\right)}$ and $\mathcal{E}_{U_i}(\pi/2)=e^{-\frac{i\pi}{2}\left(U_i(\vec{t})-I\right)}(\cdot)e^{\frac{i\pi}{2}\left(U_i(\vec{t})-I\right)}=U_i(\vec{t})(\cdot)U^{\dagger}_i(\vec{t})$, we have
\begin{eqnarray}
    \begin{split}
         \left\|\mathcal{E}_{V_i}(\pi/2)-\mathcal{E}_{U_i}(\pi/2)\right\|_{\diamond}&= \left\|\mathcal{E}_{V_i}(\pi/2J)\cdots\mathcal{E}_{V_i}(\pi/2J) -\mathcal{E}_{U_i}(\pi/2J)\cdots\mathcal{E}_{U_i}(\pi/2J) \right\|_{\diamond}\\
         &\leq J\left\|\mathcal{E}_{V_i}(\pi/2J)-\mathcal{E}_{U_i}(\pi/2J)\right\|_{\diamond}\\
         &\leq 2J\left\|e^{-\frac{i\pi}{2J}\left(V_i(\vec{t})-I\right)}-e^{-\frac{i\pi}{2J}\left(U_i(\vec{t})-I\right)}\right\|_{\infty}.
    \end{split}
\end{eqnarray}
Using the Taylor series to rewrite above operators results in
\begin{eqnarray}
    \begin{split}
         &2J\left\|\sum\limits_{k\geq 0}\frac{(-i\pi/2J)^k}{k!}\left[(V_i(\vec{t})-I)^k-(U_i(\vec{t})-I)^k\right]\right\|_{\infty}\\
         &\leq 2J\|V_i(\vec{t})-U_i(\vec{t})\|_{\infty}\sum\limits_{k\geq 0}\frac{(\pi/2J)^k}{k!}k\max\{\|V_i(\vec{t})-I\|_{\infty}^{k-1},\|U_i(\vec{t})-I\|_{\infty}^{k-1}\}\\
         &\leq \pi\|V_i(\vec{t})-U_i(\vec{t})\|_{\infty}\exp\left(\frac{\pi}{2J}\max\{\|V_i(\vec{t})-I\|_{\infty},\|U_i(\vec{t})-I\|_{\infty}\}\right)\\
         &\leq \pi e^{\pi/2}\epsilon.
    \end{split}
\end{eqnarray}
This provides a theoretical guarantee in synthesis the quantum circuit $\mathcal{E}_V(\pi/2)$ such that
\begin{align}
    \left\|\mathcal{U}-\mathcal{E}_V(\pi/2)\right\|_{\diamond}=\left\|\mathcal{U}-\mathcal{E}_{V_1}(\pi/2)\cdots\mathcal{E}_{V_n}(\pi/2)\right\|_{\diamond}\leq n\pi e^{\pi/2}\epsilon.
\end{align}

Finally, we provide the quantum circuit depth estimation to the Hamiltonian dynamics $\mathcal{E}_V(\pi/2)$. Noting that, the learned quantum circuit can be complied by the quantum dynamics
\begin{align}
   S\prod\limits_{i=1}^ne^{-i\pi/2(V_i(\vec{t})-I)}.
\end{align}
Here, each term $V_i(\vec{t})=\sum_{l=1}^L\alpha_lQ_l$ contains $3L$ Pauli terms ($L$ is given by Eq.~\ref{Eq:termupperbound}), and the support of involved Pauli terms nontrivially acts on the $i$-th qubit. This property enables that one cannot perform several $Q_l$ simultaneously. Using the $p$-th order Trotter-Suzuki method~\cite{childs2021theory}, we can approximate  $e^{-i\pi/2(V_i(\vec{t})-I)}$ by a quantum circuit with circuit depth
\begin{align}
    d=\mathcal{O}\left(\frac{\pi/2\left(\sum_{\gamma_1,\cdots,\gamma_p=1}^{3L}\left\|[Q_{\gamma_{p}}\cdots[Q_{\gamma_2},Q_{\gamma_1}]]\right\|\right)^{1/p}}{\epsilon^{1/p}}\right)\leq \mathcal{O}\left(\frac{\pi ((3L)^p2^p\max(\|Q_l\|)^p)^{1/p}}{2\epsilon^{1/p}}\right).
\end{align}

We note that $\mathcal{O}(n/M^D(t))$ local evolutions $e^{-i\pi/2(V_i(\vec{t})-I)}$ can be implemented simultaneously, as a result, the circuit depth of $\mathcal{E}_V$ can be approximated by 
\begin{align}
        \mathcal{O}\left(3M^D(t) \left[4^{K\Lambda}e\mathfrak{d}\right]^{M(t)}/\epsilon^{1/p} \right),
\end{align}
with $D$ the dimension of Hamiltonian dynamics, and $M(t)$ is given by lemmas~\ref{lemma1}.

\section{Efficient Training Algorithm for Quantum Classifier}

Specifically, Lemmas~\ref{lemma1} demonstrates that $U^{\dagger}(\vec{\bm \theta})OU(\vec{\bm \theta})$ can be approximated by a linear combination of Pauli operators $U_O=\sum_{j=1}^L\alpha_jQ_j$ with $L\leq\mathcal{O}\left(\left(e\mathfrak{d}\right)^{M(t)}\right)$, such that 
$\|U^{\dagger}(\vec{\bm \theta})OU(\vec{\bm \theta})-U_O\|\leq\epsilon$ and $\abs{{\rm supp}(Q)}\leq\mathcal{O}(\log n)$. Let $$\hat{\mathcal{L}}=\frac{1}{N}\sum_{i=1}^N\abs{\sum_{j=1}^L\alpha_j\langle\phi(x_i)|Q_j|\phi(x_i)\rangle-y_i},$$ then this directly implies
$|\mathcal{L}-\hat{\mathcal{L}}|\leq\epsilon$. 

On other hand, when $N\leq L$, we can find a solution  $\vec{\bm\alpha}=(\alpha_1,\cdots,\alpha_L)$ via solving a linear system $\bm\Phi\vec{\bm\alpha}=\vec{\bm y}$, where the matrix $\Phi=\left[\langle\phi(x_i)|Q_j|\phi(x_i)\rangle\right]_{j,i}$ and $\vec{\bm y}=(y_1,\cdots,y_N)$. Finally, the trained quantum classifier can be written by 
\begin{align}
    f(\cdot)=\sum_{j=1}^L[(\Phi^{\dagger}\Phi)^{-1}\Phi^{\dagger}\vec{\bm y}]_j{\rm Tr}(Q_j\cdot).
\end{align}

\section{Quantum Computation Verification}
\begin{proof}[Proof of Corollary~\ref{coro:2D}]
    Denote the time series $\vec{t}=(t_1,\cdots,t_K)$, a Hamiltonian dynamics operator $U(\vec{t})=\prod_{k=1}^Ke^{iH^{(k)}t_k}$ driven by Hamiltonians $\{H^{(1)},\cdots,H^{(K)}\}$ and the observable $O=O_1\otimes\cdots\otimes O_n$, the quantum dynamics mean value can be equivalently computed by
\begin{eqnarray}
    \begin{split}
         \mu(\vec{t})&=\langle0^n|U^{\dagger}(\vec{t})\left(O_1\otimes\dots\otimes O_n\right)U(\vec{t})|0^n\rangle\\
         &=
         \langle0^n|\left(U^{\dagger}(\vec{t})O_1U(\vec{t})\right)\left(U^{\dagger}(\vec{t})O_2U(\vec{t})\right)\cdots\left(U^{\dagger}(\vec{t})O_nU(\vec{t})\right)|0^n\rangle\\    
    &=\langle0^n|U_1(\vec{t})U_2(\vec{t})\cdots U_n(\vec{t})|0^n\rangle,
    \end{split}
\end{eqnarray}
where $U_i(\vec{t})=U^{\dagger}(\vec{t})O_iU(\vec{t})$.  
Alg.~\ref{Algorithm} approximates $U_i(\vec{t})$ by $V_i(\vec{t})$ such that $\|U_i(\vec{t})-V_i(\vec{t})\|\leq \mathcal{O}(\epsilon/2n)$, where $V_i(\vec{t})$ is essentially a linear combination of ${\rm poly}(n)$ matrices which nontrivial act on at most $\mathcal{O}(e^{K\mathfrak{d}t}\log(2n/\epsilon))$ qubits. As a result, the mean value $\mu(\vec{t})$ can be approximated by $ \hat{\mu}(\vec{t})=\langle0^n|V_1(\vec{t})\cdots V_n(\vec{t})|0^n\rangle$ such that $\abs{\mu(\vec{t})-\hat{\mu}(\vec{t})}\leq \epsilon/2$.

Then we follow the causality principle and the lightcone of $V_i(\vec{t})$ to assign $\{V_i(\vec{t})\}_{i=1}^n$ into two different groups, which are denoted by $V(R_1)$ and $V(R_2)$. This method is first studied in Ref.~\cite{bravyi2021classical} to simulate constant 2D digital quantum circuits, and is extended to 2D constant time Hamiltonian dynamics~\cite{wu2024efficient}. It is shown that each region ($R_1$ or $R_2$) consists of $\sqrt{n}/4M(t)$ sub-regions which are separated by $\geq 2M(t)$ distance. This property enables operators $V(R_1)$ and $V(R_2)$ are easy to simulate classically, and the quantum dynamics mean value has the form $\hat{\mu}(t)=\langle0^n|V(R_1)V(R_2)|0^n\rangle$. Then the classical Monte Carlo algorithm can be used to approximate $\hat{\mu}(t)$. Noting that operators $V(R_1)$ and $V(R_2)$ are not always unitary matrices, they have to be normalized in advance, such that $\gamma_i=\|V(R_i)|0^n\rangle\|^2\leq1$ for $i\in\{1,2\}$. This step can be implemented efficiently since both $V(R_1)$ and $V(R_2)$ are the product of some local operators $V_i(\vec{t})$ which can be normalized easily. As a result, as a mean value of $$F(x)=\frac{\gamma_1\langle x|V(R_2)|0^n\rangle}{\langle x|V^{\dagger}(R_1)|0^n\rangle}$$ with $x$ samples from $$p(x)=\gamma_1^{-1}\abs{\langle0^n|V(R_1)|x\rangle}^2,$$ we have
\begin{eqnarray}
         \hat{\mu}(t)=\sum\limits_{x}\langle0^n|V(R_1)|x\rangle\langle x|V(R_2)|0^n\rangle
         =\sum\limits_{x}p(x)\frac{\gamma_1\langle x|V(R_2)|0^n\rangle}{\langle x|V^{\dagger}(R_1)|0^n\rangle},
         \label{Eq:MCMC}
\end{eqnarray}
and the variance of $F(x)$ is given by ${\rm Var}(F)=\sum_{x}p(x)\left\|\frac{\gamma_1\langle x|V(R_2)|0^n\rangle}{\langle x|V^{\dagger}(R_1)|0^n\rangle}\right\|^2-\hat{\mu}^2(t)=\gamma_1\gamma_2-\hat{\mu}^2(t)\leq 1$. Ref.~\cite{wu2024efficient} theoretically demonstrated that computing the function $F(x)$ and probability $p(x)$ require $\mathcal{O}(n^{e^{K\pi e\mathfrak{d}t}\log(n/\epsilon)})$ classical running time.

As a result, $\mathcal{O}(4/\epsilon^2)$ samples $x$ suffice to provide an estimation to $\hat{\mu}(\vec{t})$ within $\mathcal{O}(\epsilon/2)$ additive error. Combining the above steps together, a $\epsilon$ approximation to the $K$-step quantum mean value problem is provided.
\end{proof}

\section{Benchmarking noisy quantum computation}

In the context of the NISQ era, a certain level of noise exists in the quantum circuits, making the unitary process to a CPTP map. Here, we study the robustness of the proposed learning algorithm when each quantum gate is affected by a $\gamma$-strength depolarizing channel $\mathcal{N}_i(\cdot)=(1-\gamma)(\cdot)+\gamma\frac{I}{2}{\rm Tr}(\cdot)$. It has the property that $\mathcal{N}(I)=I$ and $\mathcal{N}(P)=(1-\gamma)P$ when $P\in\{X,Y,Z\}$.

\begin{definition}[Quantum Analog Computation affected by local depolarizing Channel]
\label{Eq:channel}
We assume that the noise in the quantum device is modeled by local depolarizing channel $\mathcal{N}_i$ with strength $\gamma$. Let 
$U(\vec{t})=\prod_{k=1}^{K} e^{-iH^{(k)}t_k}=U_KU_{K-1}\cdots U_1$ the $K$-layer quantum analog computation, and let $\mathcal{N} \circ\mathcal{U}_k=\left(\otimes_{i=1}^n\mathcal{N}_i\right)\circ \mathcal{U}_k$ be the representation of a noisy circuit layer. We define the $K$-depth noisy quantum state with noise strength $\gamma$ as
\begin{align}
 \mathcal{U}_{\rm noisy}(\vec{t})=\mathcal{N}\circ \mathcal{U}_{K}\circ \mathcal{N}\circ \mathcal{U}_{K-1}\circ\cdots\circ\mathcal{N}\circ \mathcal{U}_{1}.
 \label{Eq:density_matrix}
\end{align} 
\end{definition}

\begin{corollary}[Robustness to Gate error]
  Given the noisy quantum analog circuit $\mathcal{U}_{\rm noisy}(\vec{t})$ defined as Def.~\ref{Eq:channel}, Alg.~\ref{Algorithm} may output an operator $\tilde{U}$ such that
    \begin{align}
        \left\|U(\vec{t})\otimes U^{\dagger}(\vec{t})-\tilde{U}\right\|_{\infty}\leq \gamma n^2,
    \end{align}
    where $U(\vec{t})=U_K\cdots U_1$ represents the noiseless quantum circuit related to $\mathcal{U}_{\rm noisy}(\vec{t})$. 
    \label{corollary5}
\end{corollary}

\begin{proof}
    Recall that Alg.~\ref{Algorithm} reconstructs the quantum circuit $U(\vec{t})\otimes U^{\dagger}(\vec{t})$ via the identity $S\prod_{i=1}^n\left[U^{\dagger}(\vec{t})S_iU(\vec{t})\right]$, while the operator $U^{\dagger}(\vec{t})S_iU(\vec{t})$ may become to $\mathcal{U}_{\rm noisy}(\vec{t})(S_i)$ in the noisy environment. Now we evaluate the distance $\|\mathcal{U}_{\rm noisy}(\vec{t})(S_i)-U^{\dagger}(\vec{t})S_iU(\vec{t})\|_{\infty}$ by using the Pauli path propagation method~\cite{aharonov2022polynomial}. Specifically, the normalized $n$-qubit Pauli operator $s_k\in\{I/\sqrt{2},X/\sqrt{2},Y/\sqrt{2},Z/\sqrt{2}\}^{\otimes n}$ is vectorized by $|s_k\rangle\rangle$. Using the property $I=\sum_{s_k}|s_k\rangle\rangle\langle\langle s_k|$, we have
    \begin{eqnarray}
        \begin{split}
            &\|\mathcal{U}_{\rm noisy}(\vec{t})(S_i)-U^{\dagger}(\vec{t})S_iU(\vec{t})\|_{\infty}\\
            =&\left\|\sum\limits_{s_0,s_1,\cdots,s_{K+1}}\left[(1-\gamma)^{\abs{\vec{s}}}-1\right]\langle\langle S_i|s_{K+1}\rangle\rangle\langle\langle s_{K+1}|\mathcal{U}_{K}|s_K\rangle\rangle\cdots\langle\langle s_2|\mathcal{U}_1|s_1\rangle\rangle s_0\right\|_{\infty}\\
            \leq&\max\limits_{\abs{\vec{s}}}(1-(1-\gamma)^{\abs{\vec{s}}})\|U^{\dagger}S_iU\|_{\infty}.
        \end{split}
    \end{eqnarray}
For small enough noise strength $\gamma$, above metric can be upper bounded by $\gamma n\|U^{\dagger}S_iU\|_{\infty}\leq \gamma n$. Taking this metric difference to the whole system, we finally have
\begin{align}
    \left\|U(\vec{t})\otimes U^{\dagger}(\vec{t})-\tilde{U}\right\|_{\infty}= \left\|S\prod_{i=1}^n\left[U^{\dagger}(\vec{t})S_iU(\vec{t})\right]-S\prod_{i=1}^n\mathcal{U}_{\rm noisy}(\vec{t})(S_i)\right\|_{\infty}\leq\sum\limits_{i=1}^n\left\|U^{\dagger}(\vec{t})S_iU(\vec{t})-\mathcal{U}_{\rm noisy}(\vec{t})(S_i)\right\|_{\infty}\leq \gamma n^2.
\end{align}
\end{proof}

\begin{corollary}
  Given the $n$-qubit noisy quantum analog circuit $\mathcal{U}_{\rm noisy}(\vec{t})$ defined as Def.~\ref{Eq:channel}, Alg.~\ref{Algorithm} may output a $2n$-qubit channel $\tilde{\mathcal{U}}$ such that
  \begin{align}
      \abs{{\rm Tr}\left[O\left(\mathcal{U}_{\rm noisy}(\vec{t})(|\psi\rangle\langle\psi|)\otimes \frac{I_n}{2^n}\right)\right]-{\rm Tr}\left[O\tilde{\mathcal{U}}\left(|\psi\rangle\langle\psi|\otimes \frac{I_n}{2^n}\right)\right]}\leq \mathcal{O}(\gamma n^2\|O\|_{\infty})
    \end{align}
    for any $n$-qubit quantum state $|\psi\rangle$ and $n$-qubit observable $O$. 
\end{corollary}
\begin{proof}
    Suppose $U(\vec{t})=U_K\cdots U_1$ represent the noiseless quantum circuit corresponding to $\mathcal{U}_{\rm noisy}(\vec{t})$, and $\mathcal{U}$ denotes the channel representation of $U(\vec{t})\otimes U^{\dagger}(\vec{t})$. Then the proof can be divided into two steps via using the triangle inequality. Firstly, using the Pauli path propagation method, it is shown that
    \begin{eqnarray}
    \begin{split}
   & \abs{{\rm Tr}\left[O\left(\mathcal{U}_{\rm noisy}(\vec{t})(|\psi\rangle\langle\psi|)\otimes \frac{I_n}{2^n}\right)\right]-{\rm Tr}\left[O\mathcal{U}\left(|\psi\rangle\langle\psi|\otimes \frac{I_n}{2^n}\right)\right]}\\
         =&\abs{{\rm Tr}\left[O\mathcal{U}_{\rm noisy}(\vec{t})(|\psi\rangle\langle\psi|)\otimes I_n/2^n\right]-{\rm Tr}\left[OU(\vec{t})|\psi\rangle\langle\psi|U^{\dagger}(\vec{t})\right]}\\
         \leq&\max\limits_{\abs{\vec{s}}}(1-(1-\gamma)^{\abs{\vec{s}}})\abs{{\rm Tr}\left[OU(\vec{t})|\psi\rangle\langle\psi|U^{\dagger}(\vec{t})\right]}\\
         \leq& \gamma n\|O\|_{\infty},
    \end{split}
    \end{eqnarray}
    where the third line comes from the Pauli path propagation. 
    
    Secondly, for any operator $A$, $B$, quantum state $|\psi\rangle$, and observable $O$, we have the relationship
    \begin{eqnarray}
    \begin{split}
        &\abs{\langle\psi|A^{\dagger}OA|\psi\rangle-\langle\psi|B^{\dagger}OB|\psi\rangle}\\
        \leq& \sqrt{\langle\psi|\left(A^{\dagger}OA-B^{\dagger}OB\right)^2|\psi\rangle}\\
        \leq&  \left\|(A^{\dagger}-B^{\dagger})OA+B^{\dagger}O(A-B)\right\|_{\infty} \\
        \leq &2\|O\|_{\infty}\max\{\|A\|_{\infty},\|B\|_{\infty}\}|\|A-B\|_{\infty},
    \end{split}
\end{eqnarray}
where the second line comes from the Cauchy inequality. Recall that $\tilde{U}=S\prod_{i=1}^n\mathcal{U}_{\rm noisy}(\vec{t})(S_i)$, and this results in
\begin{align}
    \|\tilde{U}\|_{\infty}\leq \prod_{i=1}^n\|\mathcal{U}_{\rm noisy}(\vec{t})(S_i)\|_{\infty}\leq (1-\gamma)^n<1.
\end{align}

Assign $\mathcal{U}(\cdot)=A(\cdot)A^{\dagger}$, $\tilde{\mathcal{U}}=B(\cdot)B^{\dagger}$ and $2n$-qubit density matrix $\rho=|\psi\rangle\langle\psi|\otimes I_n/2^n$, we have
    \begin{eqnarray}
        \begin{split}
             \abs{{\rm Tr}\big[O\mathcal{U}(\rho)\big]-{\rm Tr}\left[O\tilde{\mathcal{U}}(\rho)\right]}
             \leq  2\|O\|_{\infty}\left\|U(\vec{t})\otimes U^{\dagger}(\vec{t})-\tilde{U}\right\|_{\infty}\leq 2\gamma n^2\|O\|_{\infty},
        \end{split}
    \end{eqnarray}
where the last inequality comes from Corollary~\ref{corollary5}. Combining above two results by using the triangle inequality, we complete the proof.
\end{proof}

\end{document}